\documentclass[11pt, a4paper]{article}

\RequirePackage{amsthm,amsmath,amsfonts,amssymb}
\RequirePackage[authoryear]{natbib}
\RequirePackage[colorlinks,citecolor=blue,urlcolor=blue]{hyperref}
\RequirePackage{graphicx}
\usepackage{booktabs}
\usepackage{multirow}
\usepackage{algpseudocode}
\usepackage[ruled]{algorithm2e}
\usepackage{subcaption} 
\usepackage{dsfont}
\usepackage{geometry}
\usepackage{authblk}

\theoremstyle{plain}

\newtheorem{theorem}{Theorem}[section]
\newtheorem{lemma}[theorem]{Lemma}
\newtheorem{proposition}[theorem]{Proposition}

\theoremstyle{definition}
\newtheorem{definition}[theorem]{Definition}
\newtheorem*{example}{Example}

\newtheorem*{remark}{Remark}

\begin{document}

\title{Credible rectangles for high-dimensional posterior comparison}

\author[1]{Alice Chevaux}
\author[1]{Julyan Arbel}
\author[2]{Guillaume Kon Kam King}
\author[1]{Sophie Achard}
\affil[1]{Univ. Grenoble Alpes, CNRS, Inria, LJK, F-38000 Grenoble, France}
\affil[2]{Université Paris-Saclay, INRAE, MaIAGE, 78350 Jouy-en-Josas, France}
\date{} 

\maketitle

\begin{abstract}
We propose a Bayesian framework for uncertainty quantification and comparison in brain connectivity graph analysis. Standard graph-based approaches typically rely on point estimates of correlation matrices, overlooking the uncertainty induced by high-dimensional estimation from limited data. Our methodology constructs and compares credible hyperrectangles derived from posterior distributions, providing interpretable tools for subject-level inference and longitudinal monitoring. We develop scalable algorithms for estimating these regions in high dimensions and establish theoretical guarantees in the inverse-Wishart model for resting-state fMRI data, including a Bernstein--von Mises theorem for correlation matrices and control of a Bayesian family-wise error rate. The proposed framework enables principled detection of significant connectivity differences both globally and locally while preserving joint dependency structures. While demonstrating competitive performance against multiple-testing procedures on synthetic datasets, our approach also facilitates the direct comparison of two distinct scans from a single patient, a capability currently absent from the literature. We leverage this novelty on real datasets to improve interpretability. Beyond fMRI data, the approach provides a general framework for comparison problems in high-dimensional dependent settings.
\end{abstract}

\section{Introduction}

Understanding and comparing brain connectivity at the individual level is a central challenge in modern neuroimaging. While graph-based representations derived from  functional Magnetic Resonance Imaging (fMRI) data have become standard, they often rely on point estimates that overlook the uncertainty inherent in their construction. In this work, we propose a Bayesian framework to quantify and propagate this uncertainty, enabling principled comparisons between connectivity patterns at both the individual and group levels. Crucially, our methodology is designed to be adaptable across various models, addressing a critical unmet clinical need: the absence of existing tools for robust decision-making when comparing single-subject scans, either against one another or against a reference value for instance derived from healthy control subjects. We believe this approach significantly enhances interpretability for physicians, directly translating complex statistical uncertainties into actionable clinical insights.

\subsection{Motivations}

Graph theory has emerged as a well-established framework for modeling and interpreting the complex connectivity patterns inherent in brain imaging data. Yet, the construction of representative brain graphs remains a delicate, multi-step procedure. The standard pipeline typically entails estimating a correlation matrix from the timeseries obtained with preprocessed imaging signals, followed by the application of a graph estimation model. While numerous estimators have been developed to stabilize the correlation matrix, the variance intrinsic to this estimation step is often underestimated or neglected when the matrix is subsequently treated as a graph.

This statistical limitation is exacerbated by the escalating heterogeneity and scale of contemporary neuroimaging datasets and the lack of standardized preprocessing procedures. Modern studies must account for variable acquisition durations, diverse scanner hardware, distinct noise profiles, and expanding population cohorts. Moreover, the increasing accessibility of longitudinal data has elevated the importance of follow-up analyses for clinical applications, facilitating both group-level comparisons and individual patient monitoring. Consequently, whether for comparative studies or longitudinal tracking, there is a critical need to quantify the uncertainty propagating through the acquisition and analysis pipeline from a purely statistical standpoint, independent of the system's physical constraints.

To gain a clearer understanding of the uncertainty that occurs and propagates during the analysis of brain imaging data, we adopt a Bayesian framework. The Bayesian approach is particularly well-suited for individual-level inference, as it enables the direct quantification and propagation of uncertainty in non-asymptotic settings, while naturally preserving the joint dependency structure intrinsic to high-dimensional brain connectivity networks. Key to our approach is the exploitation of a simple Bayesian model using credible regions. Contributions include a general methodology leveraging credible hyperrectangles to furnish statistical tools for comparing two brain connectivity graphs and enhancing interpretability at the subject level. We further provide algorithms to compute these credible hyperrectangles in high-dimensional settings. Specifically, within the context of an inverse-Wishart model for resting-state fMRI (rs-fMRI) data, we establish theoretical result, including a Bernstein--von Mises theorem and a procedure to control a Bayesian analogue of the Family-Wise Error Rate (FWER).

Although this work primarily focuses on rs-fMRI data due to its prominence in connectivity research, the proposed methodological framework is broadly applicable. Most of our findings generalize to other high-dimensional neuroimaging modalities, such as electroencephalography (EEG), which face analogous challenges in estimating stable correlation structures from noisy, high-dimensional time series.

\subsection{Data specificities}

This paper concerns the analysis of rs-fMRI datasets. While we do not discuss preprocessing or parcellation of the brain, we focus on wavelet decomposition of the signal \citep{whitcher_wavelet_2000}. The properties of wavelets provide a context in which assuming that the observations are independent and identically distributed (i.i.d.) is a reasonable hypothesis, as they erase short-range dependencies \citep{achard_multivariate_2016}. The choice of wavelet scale is made to ensure that frequencies remain below 0.1 Hz, as resting-state oscillations relevant to functional connectivity are considered most salient within this low-frequency range \citep{biswal_functional_1995,fox_spontaneous_2007,achard_efficiency_2007}.

We have chosen to study datasets which have already been analysed using wavelet analysis. Some conclusions have already been drawn from these datasets at a global level (group comparisons), but our aim is to assess significant results at an individual level.

The first dataset is the Human Connectome Project (HCP) test-retest database \citep{smith_resting-state_2013,glasser_minimal_2013} that has been studied by \cite{termenon_reliability_2016}. It consists of $100$ healthy volunteers. Each subject underwent two rs-fMRI acquisitions on different days. This dataset is useful for assessing the reliability of the method: for a reliable method, we would expect similar results when studying the same subject. To provide a point of reference, the similarity within a subject is compared to the similarity between two different subjects. This within-subject and between-subject comparison determine whether a method can be used to characterise a subject. Preprocessing is described in \cite{termenon_reliability_2016}, and we focus on scale 4 of wavelets to obtain the desired frequency. The parcellation using the modified anatomic-automatic labeling (AAL) atlas \citep{tzourio-mazoyer_automated_2002} gives us 89 Regions of Interest (ROI) to represent the brain, and the number of observations is $n = 68$. 

The second dataset, which is the main focus of this project, consists of scans of patients who have experienced a brain injury resulting in a coma. Patients were recruited in the ICU at Grenoble-Alpes University Hospital; ten of them were screened twice: once 30 days and once 60 days after the brain injury. We focus on these $10$ patients and $10$ healthy volunteers scanned in the same set-up to study the recovery of the patients after brain injury. Details of the preprocessing and parcellation choices are available in \cite{oujamaa_functional_2023}. This time, there are $107$ regions of interest (ROIs) using modified Anatomic-Automatic-Labelling 3 (AAL3) \citep{rolls_automated_2020} and scale 3 of the wavelet decomposition was used which gives us $95$ observations per ROI. We also had access to ten control volunteers who underwent the same protocol, providing data on controls under the same conditions. This raises several clinical questions: can we characterise the evolution of a patient? Given the variety of brain injuries, is it possible to provide individual follow-up? Is there a way to identify remission? These clinical questions form the basis of this work. 

\subsection{Contributions and organization}
To address the lack of uncertainty-aware subject-level comparison tools in high-dimensional connectivity analysis, we propose a general Bayesian framework based on credible rectangles derived from posterior distributions. Although motivated by rs-fMRI connectivity analysis, the methodology is broadly applicable to high-dimensional comparison problems involving dependent parameters.

Our main contributions are fourfold:

\noindent\textit{Construction of credible rectangles in high dimension.}
    We introduce quantile-based credible hyperrectangles adapted to high-dimensional posterior distributions and study their statistical properties for uncertainty quantification and decision-making.
    
\noindent\textit{A comparison framework for posterior objects.}
    We develop global and local comparison procedures based on overlap properties of credible rectangles, enabling interpretable subject-level inference and longitudinal analysis.
    
\noindent\textit{Scalable algorithms for high-dimensional settings.}
    We propose several algorithms for estimating credible rectangles under severe computational and memory constraints, including online and sliced-quantile strategies designed for dimensions encountered in neuroimaging applications.
    
\noindent\textit{Theoretical guarantees and Bayesian error control.}
    In the inverse-Wishart model for correlation matrices, we establish a Bernstein--von Mises theorem for correlations, consistency results for support recovery, and control of a Bayesian analogue of the family-wise error rate (BFWER).

The remainder of the paper is organized as follows. Section~\ref{sec:credibleregions} introduces the general framework for credible rectangles and posterior comparison. Section~\ref{sec:correlation} specializes the methodology to correlation matrices and establishes the associated theoretical guarantees. Section~\ref{sec:resultssimu} evaluates the proposed procedures on synthetic datasets, while Section~\ref{sec:realdatasets} presents applications to longitudinal and subject-level rs-fMRI connectivity analysis.

All proofs are deferred to Appendix, code is available on \hyperlink{https://gricad-gitlab.univ-grenoble-alpes.fr/chevaual/credible_rectangles_for_comparison_fmri}{Gitlab}.

\newpage
\section{Related works in fMRI study}

\paragraph*{Group comparison}

Several methods have been proposed for comparing groups of patients across various pathologies. Most approaches rely on statistical tests applied to specific terms or network metrics, such as multiple testing procedures on correlation matrix coefficients or graph-theoretical measures \citep{kim_comparison_2014, zalesky_network-based_2010}. Other studies, notably \cite{ginestet_hypothesis_2017}, have advanced hypothesis testing at the global graph level, with potential post-hoc interpretation at the edge level. However, to achieve sufficient statistical power, these group-centric frameworks inherently rely on a more or less explicit assumption of homogeneity within the studied populations. This requirement poses a significant limitation in clinical contexts where patient variability is high. Indeed, the presence of heterogeneous subgroups or individual outliers can dilute group-level effects, leading to reduced sensitivity and potentially explaining the modest performance of such methods in certain studies. By averaging signals across subjects, these approaches inevitably overshadow individual specificities, making it difficult to derive biomarkers that are representative of single patients. Consequently, we argue that shifting the paradigm towards subject-level methods is essential. Such an approach would not only better account for group heterogeneity and robustness against outliers but also provide the granular tools necessary for precise individual patient monitoring and personalized clinical follow-up.

\paragraph*{Continuous individual models}
At the individual level, continuous models primarily serve as a robust preprocessing step rather than a tool for direct interpretation. Their main objective is to provide a stable estimator of the correlation matrix, significantly reducing the estimation variance inherent to empirical correlation matrices, especially when the number of time points is limited relative to the number of regions. This shrinkage towards a structured target is achieved through frequentist approaches, such as the linear shrinkage estimator proposed by \cite{ledoit_honey_2003}, or within a Bayesian framework by specifying an inverse-Wishart prior on the covariance matrix, often centered on the identity matrix \citep{gelman_bayesian_1995}. While these methods yield full matrices that can be used to sample posterior distributions of transformations like partial correlations \citep{marrelec_using_2007}, the resulting dense connectivity patterns remain difficult to interpret biologically without further thresholding or processing.

\paragraph*{Sparse individual models}
Conversely, sparse models aim to enhance interpretability by explicitly identifying a subset of non-zero edges, effectively defining the support of the connectivity network. While a substantial body of literature exists on sparse graphical models based on the precision matrix, our focus here remains strictly on the correlation structure, thus excluding precision-based approaches. Within the realm of sparse correlation estimation, two distinct strategies emerge. The first relies on multiple testing procedures to determine the network support; methods controlling the Family-Wise Error Rate (FWER), such as Bonferroni \citep{bonferroni_teoria_1936} or Holm--Bonferroni \citep{holm_simple_1979} corrections, offer rigorous statistical control over false positives but can be considered too conservative in practice. The second strategy involves direct estimation of a sparse correlation matrix, ensuring the semi-definite positive property of the estimator \citep{bien_sparse_2011}. However, these penalized estimation techniques do not offer dedicated tools for comparing two individuals, and they rely on a strong sparsity assumption that may not always reflect the complexity of real-world biological networks, potentially overlooking weak yet meaningful connections as it was shown in previous benchmark \citep{chevaux_benchmarking_2025}.

Ultimately, while these individual-level methods effectively address covariance estimation and sparsity, they fail to provide statistical tools for directly comparing two scans from the same patient. This gap in the literature precludes rigorous longitudinal analyses and the detection connectivity changes at the subject level.

\section{Credible rectangle for comparisons problems} \label{sec:credibleregions}

In this section, we consider the multivariate setting where $\theta = (\theta_1,\dots,\theta_d) \in \mathbb{R}^d$, temporarily setting aside the specific matrix constraints that will be addressed in Section~\ref{sec:correlation}. We assume a Bayesian framework characterized by a prior distribution $\pi(\theta)$ and a posterior distribution $\pi(\theta|\mathbf{y})$ for the parameter of interest. To quantify the uncertainty associated with a Bayesian estimator, it is necessary to summarize this $d$-dimensional distribution. A standard approach is to construct a credible region $\mathcal{R} \subset \mathbb{R}^d$. For a given significance level $\alpha \in (0,1)$, a credible region of level $1-\alpha$ is defined as a subset satisfying:

\[
\mathbb{P}(\theta \in \mathcal{R} \mid \mathbf{y}) = \int_{\mathcal{R}} \pi(\theta|\mathbf{y}) \, d\theta \geq 1 - \alpha.
\]

While working in a Bayesian framework typically provide ways to sample easily from the posterior distribution, defining such a region in high dimensions remains a non-trivial challenge. The following section explores various constructions of multidimensional credible regions and their specific properties. This problem bears some similarity with the construction of multivariate confidence regions in frequentist statistics. However, whereas frequentist methods often rely on asymptotic normality or restrictive parametric assumptions about the estimator's distribution, the Bayesian approach leverages the full posterior distribution directly. This allows for uncertainty quantification under the assumed model for finite samples, without resorting to asymptotic approximations.

\subsection{Credible region for high-dimensional parameters}

\subsubsection{State of the art}

From a Bayesian perspective, the construction of credible regions is largely divided into two types: highest posterior density (HPD) regions and quantile-based intervals. The former are favoured since they offer the smallest volume for a given credible level, but computing them becomes intractable in high dimensions as kernel-based methods typically fail when $d > 100$ \citep{ferrie_high_2014}. In the context of neuroscience, where dimensions reach approximately $5000$, this forces the adoption of regions parameterised by a finite number of parameters, specifically the Cartesian product of quantile-based intervals. Often termed simultaneous intervals \citep{joshi_bayesian_2023} or bands \citep{held_simultaneous_2004}, these structures are increasingly recognized as necessary for managing joint uncertainty; recent post-processing approaches \citep{griffin_expressing_2025} validate this reliance on Cartesian products (or "rectangles") for visualization and decision-making. However, while existing methods leverage this structure primarily for variable selection using $50\%$ credible sets as an alternative to Posterior Inclusion Probabilities, our work addresses medical applications and global risk control. In these contexts, rigorous multiple testing corrections mandate high-coverage regions ($\geq 95\%$), rendering rectangular approximations a computational imperative rather than a mere simplification for $d \approx 5000$, overcoming the memory limitations that restrict prior simultaneous interval methods.

From a frequentist perspective, this type of uncertainty quantification can also be applied to the estimator using confidence regions. The type of region is also restrained and tends to produce similar results to the Bayesian regions, for instance with simultaneous confidence intervals \citep{tang_simultaneous_2020}. Other interval products are proposed in \cite{chang_confidence_2018}, but with intervals of the same length, which differs from quantile-based intervals.

Overall, none of these methods are computationally feasible for dimensions larger than $100$.

\subsubsection{Criteria for credible region}
Given the state of the art and the application context, we choose to work with a Cartesian product of quantile-based intervals, which we call a \textit{quantile-based rectangle}. This gives us a parameterised region that depends only on the quantile level of each dimension, making estimation feasible even in high dimensions. The choice of a rectangular shape is not optimal in terms of volume, since an ellipsoid would capture a similar posterior mass with a lower volume. However this choice allows to assess the effect of each dimension, which is useful for interpretation as we demonstrate later. Working with an ellipsoid region would complexify the parametrisation just to be once again projected as an interval on each dimension for the interpretation. For these reasons, estimating rectangle regions directly is more straightforward.

Reservations about the use of quantile-based intervals still apply: this type of region is strongly discouraged if the distribution is multimodal. However, in the case of a unimodal distribution, this could lead to good results. See \citet{hopfl_marginal_2024} for criteria under which the HPD region can resemble a quantile-based region. It could be possible to compute the HPD region and use its marginal projection to apply the methodology proposed in Section~\ref{sec:interpretability}. However, this will ultimately result in an increase of volume.

To simplify the parametrisation further, we choose to have a \textit{single quantile level} $t$ for each interval. This is particularly sensible for our neuroscience application, as we do not want to differentiate the brain regions a priori.

\begin{definition}[Quantile-Based Rectangle $\&$ Credible Rectangle]
Let $\theta=(\theta_1, \dots, \theta_{d})$ be the parameter of interest. For each $\theta_i$, let $q_i(t)$ be the quantile of level $t$ for the marginal posterior distribution of $\theta_i$. The corresponding \textit{quantile-based rectangle} at the level $t$ is defined as:
\[
R(t)= \prod_{i=1}^{d}\left[q_{i}\left(\frac{t}{2}\right),q_{i}\left(1-\frac{t}{2}\right)\right]. \]

Such a region is called a \textit{credible rectangle} if, for some $t=t_{\alpha}$, it also verifies:
\[\mathbb{P}(\theta \in R(t_{\alpha})|\mathbf{y})\geq 1-\alpha.\]
\end{definition}

The existence and estimation of such rectangles are discussed Section~\ref{sec:construction}. 

\subsection{Methodological versatility of credible rectangles}
\label{sec:interpretability}

The originality of this work lies in leveraging credible regions as a unified framework for statistical comparison at the individual level. In this section, we develop two distinct methodological approaches: first, comparing an estimated parameter vector $\theta$ against a fixed reference $\theta_{\text{ref}}$ such as a population template, and second, directly comparing two parameter vectors $\theta_A$ and $\theta_B$ derived from distinct datasets (e.g., longitudinal scans of the same subject or data from two different subjects). While the conceptual basis of these comparisons applies broadly to credible sets, their practical implementation and interpretation in high-dimensional settings relies on the rectangular shape. Consequently, some methods presented herein are specifically tailored for credible rectangles.

\subsubsection{Comparison to a reference parameter}
\label{sec:comparaisonreference}

We first address the fundamental question of whether the parameter of interest $\theta$ deviates globally from a reference value $\theta_{\text{ref}}$. The following proposition establishes some links between the probability of $\theta$ and $\theta_{\text{ref}}$ being close if $\theta_{\text{ref}}$ does not belong to the credible rectangle of $\theta$.

\begin{proposition}[Global difference]
\label{prop:comparisonthetaref}
Let $R_{\theta} \subset \mathbb{R}^d$ be a credible rectangle for $\theta$ with credibility level $1-\alpha$, satisfying $\mathbb{P}(\theta \in R_{\theta} \mid \mathbf{y}) \geq 1-\alpha$. Let $\theta_{\text{ref}} \in \mathbb{R}^d$ be a reference value such that $\theta_{\text{ref}} \notin R_{\theta}$. Let $\| \cdot \|$ be any norm on $\mathbb{R}^d$. Define the minimal distance $d$ between $\theta_{\text{ref}}$ and $R_{\theta}$ with respect to this norm as:
\[
r = \textup{dist}(\theta_{\text{ref}}, R_{\theta}) = \inf_{\theta' \in R_{\theta}} \| \theta' - \theta_{\text{ref}} \|.
\]
Then, $r$ is positive and the posterior probability that $\theta$ lies strictly closer to $\theta_{\text{ref}}$ than $r$ is bounded by $\alpha$:
\[
\mathbb{P}(\| \theta - \theta_{\text{ref}} \| < r \mid \mathbf{y}) \leq \alpha.
\]
\end{proposition}

This proposition establishes that if the reference value is excluded from the credible region for $\theta$, we have a control over the probability of $\theta$ being close to $\theta_{\text{ref}}$. Thereby it validates the use of $\theta_{\text{ref}}$ being inside a credible region as a rigorous decision rule for global difference. Framing the decision in terms of closeness (i.e., $\| \theta - \theta_{\text{ref}} \| < r$) rather than strict equality ($\theta = \theta_{\text{ref}}$) resolves a fundamental issue in continuous inference that is that the probability of equality is zero, rendering direct probability statements about equality uninformative. By shifting the focus to whether $\theta$ lies within a $r$-neighborhood of the reference, we define a non-trivial event with positive measure, providing a meaningful quantification of difference that avoids the degeneracy of point-null hypotheses in continuous spaces. This also provide an indicator $r$ that quantify the separation of $\theta_{\text{ref}}$ from $\theta$. This result can be extended to any closed credible region.

While the global test confirms the presence of an anomaly, it does not identify its specific sources, to enable precise biological interpretation, we now decompose the decision into component-wise decision of differences. We aim to construct a decision vector $\delta \in \{0,1\}^d$ characterizing significant deviations from $\theta_{\text{ref}}$ in each dimension. 

Directly testing for equality ($\theta_i = \theta_{\text{ref},i}$) is ill-posed in a continuous framework, as the posterior probability of exact equality is zero. To circumvent this, we formulate the problem as a signed decision process. We introduce two distinct decision vectors, $\delta^+$ and $\delta^-$, both in $\{0,1\}^d$:
\begin{itemize}
    \item $\delta_i^+ = 1$ indicates a significant positive shift (i.e., we conclude $\theta_i > \theta_{\text{ref},i}$).
    \item $\delta_i^- = 1$ indicates a significant negative shift (i.e., we conclude $\theta_i < \theta_{\text{ref},i}$).
\end{itemize}

We define the loss functions associated with these directional decisions. A loss occurs if we declare a directional difference that is contradicted by the true parameter value. Specifically, for the positive decisions:
\[
L^+(\theta, \delta^+) = \mathds{1}\left\{ \bigcup_{i=1}^d \left( \{ \delta_i^+ = 1 \} \cap \{ \theta_i \leq \theta_{\text{ref},i} \} \right) \right\},
\]
which equals 1 if at least one false positive increase is declared (claiming $\theta_i > \theta_{\text{ref},i}$ when in fact $\theta_i \leq \theta_{\text{ref},i}$), and 0 otherwise. Similarly, for the negative decisions:
\[
L^-(\theta, \delta^-) = \mathds{1}\left\{ \bigcup_{i=1}^d \left( \{ \delta_i^- = 1 \} \cap \{ \theta_i \geq \theta_{\text{ref},i} \} \right) \right\}.
\]

The overall decision vector for significant differences (two-sided) is defined as the union of the directional decisions. Let $\delta_i = \max(\delta_i^+, \delta_i^-)$. The global loss function, representing the occurrence of any directional error, is:
\[
L(\theta, \delta) = \max\left( L^+(\theta, \delta^+), L^-(\theta, \delta^-) \right).
\]
In the Bayesian framework, controlling the posterior expected loss $\mathbb{E}_{\theta \mid \mathbf{y}}[L(\theta, \delta)] \leq \alpha$ ensures that the probability of making at least one erroneous directional claim s bounded by $\alpha$.

\begin{proposition}[Local differences]
\label{prop:localdiff}
Let $R(t) = \prod_{i=1}^d [q_i(t/2), q_i(1-t/2)]$ be a credible rectangle constructed from marginal posterior quantiles, with global credibility level $1-\alpha$, i.e., $\mathbb{P}(\theta \in R(t) \mid \mathbf{y}) \geq 1-\alpha$.
Let $\theta_{\text{ref}} \in \mathbb{R}^d$ be a reference value. We define the global decision vector $\delta \in \{0,1\}^d$ as:
\[
\forall i, \delta_i = \mathds{1}\left\{ \theta_{\text{ref},i} \notin \left[q_i\left(\frac{t}{2}\right), q_i\left(1-\frac{t}{2}\right)\right] \right\}.
\]
This decision $\delta_i$ corresponds to the union of two directional decisions: $\delta_i^+ = \mathds{1}\{ \theta_{\text{ref},i} < q_i(t/2) \}$ and $\delta_i^- = \mathds{1}\{ \theta_{\text{ref},i} > q_i(1-t/2) \}$. Then the loss defined earlier can be used and this procedure controls the posterior expected loss at level $\alpha$:
\[
\mathbb{E}_{\theta \mid \mathbf{y}}[L(\theta, \delta)] \leq \alpha.
\]
\end{proposition}

Although $\delta_i$ is defined as a simple test of difference, it inherently controls the directionality as well: if $\theta_{\text{ref},i} < q_i(t/2)$, we infer a significant increase from $\theta_{\text{ref},i}$ and if $\theta_{\text{ref},i} > q_i(1-t/2)$, we infer a significant decrease.

\subsubsection{Comparison between two parameters}
\label{sec:comparaisontwoparameters}

Having established a rigorous framework for comparing a single parameter $\theta$ against a fixed reference, we now extend this methodology to the direct comparison of two distinct parameter vectors, $\theta_A$ and $\theta_B$ (e.g., longitudinal scans of the same patient or data from two different individuals).

\begin{proposition}
\label{prop:comparisontworegions}
Let $R_A$ and $R_B$ be $(1-\alpha)$ rectangle credible regions for $\theta_A$ and $\theta_B$ respectively. Then if $R_A \cap R_B = \emptyset$ the probability of equality of the two parameters is controlled at a level $2 \alpha$:
\[ \mathbb{P}(\theta_A = \theta_B|\mathbf{y})\leq 2 \alpha. \]
\end{proposition}

Analogous to the reference comparison case, the global test of difference between $\theta_A$ and $\theta_B$ can be decomposed into component-wise decisions of difference by checking if intervals overlap for each dimension. Disjoint intervals are seen as significant differences, and the relative position tell us if it was a significant increase or decrease.

\subsection{Constructing credible rectangle}
\label{sec:construction}
The objective in this section is to construct a credible rectangle. For this the first step is to find $t$ such that: 
\begin{align}
    \mathbb{P}(\theta \in R(t)|\mathbf{y})\geq 1-\alpha.
    \label{eq:rectangle}
\end{align}

\subsubsection{Quantile level choice}
Certain properties of the region $R(t)$ can be derived for specific values of $t$ without imposing further assumptions on the prior distribution of the parameters. To establish a framework for selecting an appropriate quantile level, we first define three specific types of rectangles.

\begin{definition}[Special Rectangles]
We define the following three quantile levels and their corresponding rectangular regions:

(i) The \textit{Bonferroni-type rectangle}, denoted as $R(t^{\text{bonf}}_{\alpha})$, is based on the quantile level:
\[
t^{\text{bonf}}_{\alpha} = \frac{\alpha}{d}.
\]

(ii) The \textit{Naive rectangle} at a level $1-\alpha$, denoted as $R(t^{\text{naive}}_{\alpha})$, is defined by the quantile level:
\[
t^{\text{naive}}_{\alpha} = \alpha.
\]

(iii) The \textit{Optimal rectangle} at level $1-\alpha$, denoted as $R(t_{\alpha})$, corresponds to the optimal quantile level $t_{\alpha}$ satisfying:
\[
\mathbb{P}\left( \theta \in R(t_{\alpha}) \mid \mathbf{y} \right) = 1 - \alpha.
\]
\end{definition}

Having established these definitions, we can now characterize the coverage properties of the Bonferroni and Naive approaches, which subsequently guarantees the existence of the optimal level.

\begin{proposition}
\label{prop:specialrectangles}
Without further assumptions on the posterior distribution $\pi(\theta \mid \mathbf{y})$, the following properties hold:

(i) The Bonferroni-type rectangle satisfies the coverage requirement:
\[
\mathbb{P}\left( \theta \in R(t^{\text{bonf}}_{\alpha}) \mid \mathbf{y} \right) \geq 1-\alpha.
\]

(ii) The Naive rectangle satisfies:
\[
\mathbb{P}\left( \theta \in R(t^{\text{naive}}_{\alpha}) \mid \mathbf{y} \right) \leq 1-\alpha.
\]

(iii) Consequently, the optimal quantile level $t_{\alpha}$ exists and lies within the interval:
\[
t_{\alpha} \in \left[\frac{\alpha}{d}, \alpha\right].
\]
\end{proposition}

The terminology in point (i) draws from the multiple-testing correction procedure proposed by \cite{bonferroni_teoria_1936}. The existence of the optimal level stated in point (iii) follows directly from the bounds established in (i) and (ii). Since the probability mass of the quantile-based region $R(t)$ is a monotonic function of $t$ (increasing as $t$ decreases), it ensures there exists a unique $t_{\alpha}$ within the specified interval such that the coverage probability is exactly $1 - \alpha$. Detailed proofs are provided in the Appendix.

An explicit form for the optimal rectangle $R(t_{\alpha})$ is usually not available. The remaining objective is to estimate such a rectangle. Since the parametrisation depends solely on the scalar $t$ which is bounded, numerical exploration is straightforward. However, estimating extreme quantiles remains computationally prohibitive, particularly due to the substantial memory required to store sufficient Monte Carlo samples. This challenge is addressed in Section~\ref{sec:computation}.

\paragraph*{Representation of the type of rectangles}

To illustrate the limitations of Naive and Bonferroni rectangles, as well as the behaviour of the optimal rectangle,  toy distributions and projections in two dimensions can be used.

\begin{example}
The first toy distribution is a multivariate normal distribution:
\begin{align}
\label{eq:sigmarho}
    \theta|\mathbf{y} \sim \mathcal{N}_d\left(\mathbf{0}, \mathbf{\boldsymbol{\Sigma}}(\rho) \right), \text{ with } \quad \mathbf{\boldsymbol{\Sigma}}(\rho)= \rho \mathbf{1}_d+(1-\rho)I_d
\end{align}

 where $\mathbf{1}_d$ is a $d \times d$ matrix of $1$ and $I_d$ is the identity matrix.
\end{example}

\begin{example}
The second toy distribution is a distribution on correlation matrices. The parameter of the inverse-Wishart distribution is a covariance matrix, which is then transformed into a correlation matrix.

\begin{align}
\label{eq:inversewishartexemple}
    \boldsymbol{\Sigma} |\mathbf{y} \sim \mathcal{IW}_p(\mathbf{\boldsymbol{\Sigma}}(\rho),\nu=50+p+2), \quad \theta =  \boldsymbol{D}^{-\frac{1}{2}} \boldsymbol{\Sigma} \boldsymbol{D}^{-\frac{1}{2}} \quad \text{with } \boldsymbol{D}=\textup{diag}(\boldsymbol{\Sigma}).
\end{align}
\end{example}

\begin{figure*}[ht!]
    \centering
    \hspace{-1.2cm} 
    \begin{subfigure}[t]{0.49\textwidth}
        \centering
        
        \vspace{-0.2cm}
        \begin{subfigure}[t]{3.5cm}
            \centering
            \includegraphics[width=\linewidth, trim=15 15 15 15, clip]{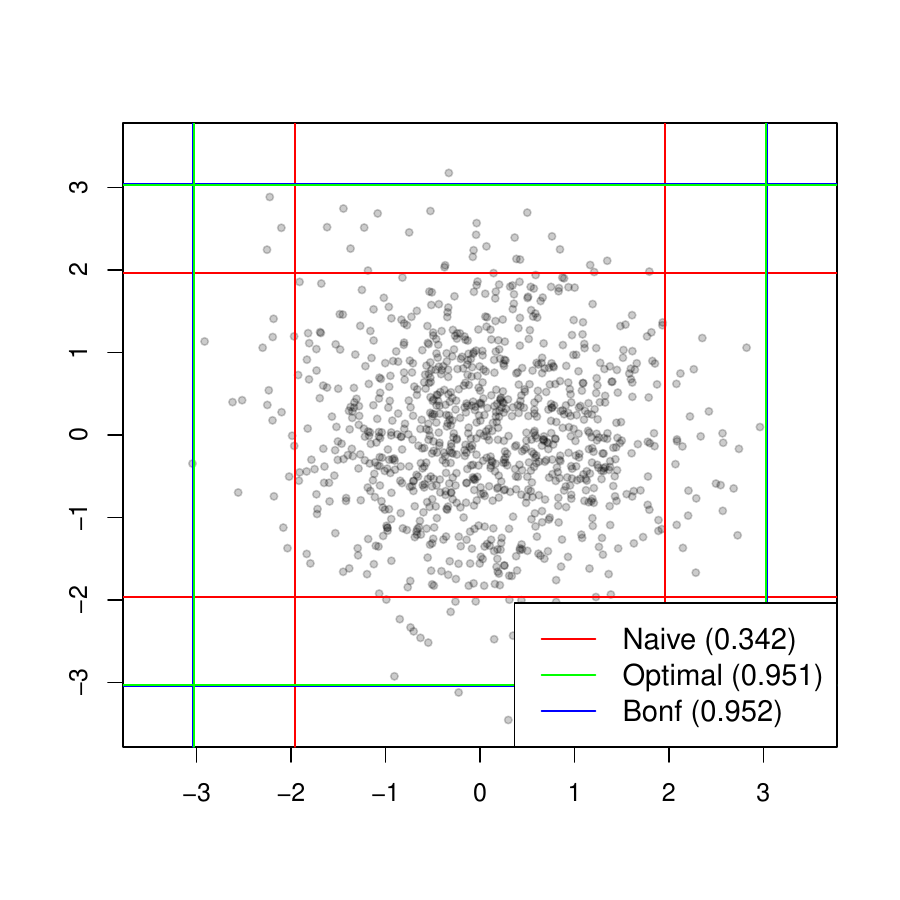}
        \end{subfigure}%
        \hspace{-0.25cm} 
        \begin{subfigure}[t]{3.5cm}
            \centering
            \includegraphics[width=\linewidth, trim=15 15 15 15, clip]{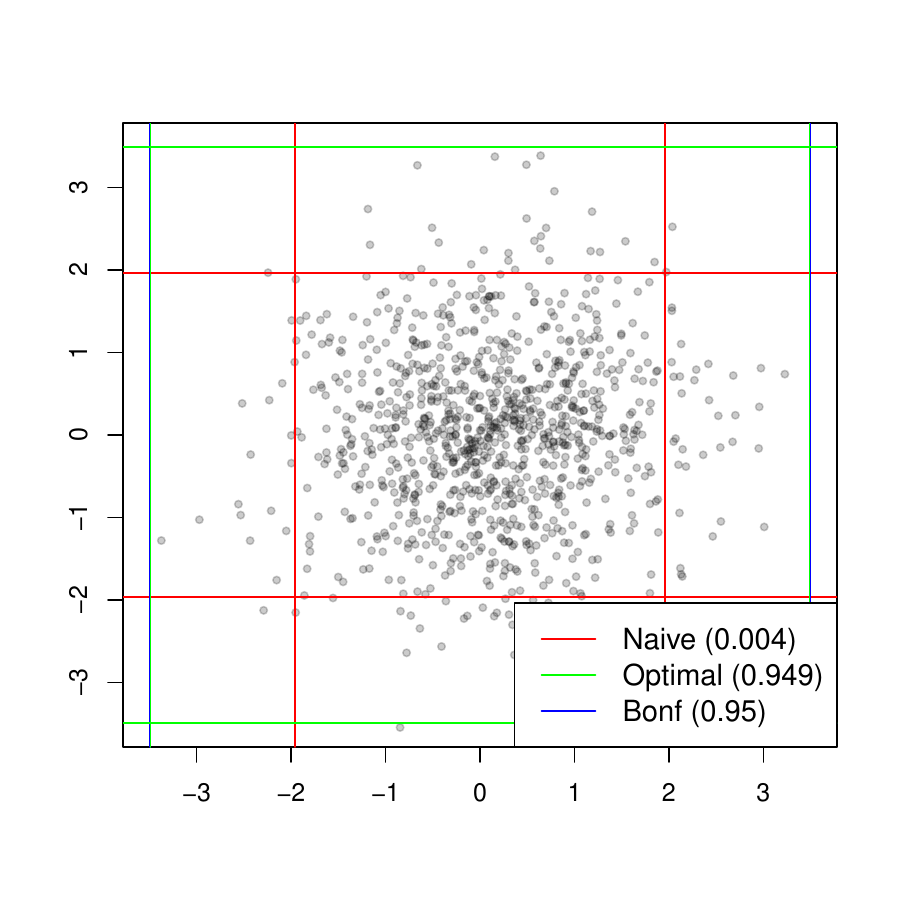}
        \end{subfigure}
        
        \vspace{-0.35cm} 
        \begin{minipage}{3.5cm}
            \centering\scriptsize $d=21, \rho=0$
        \end{minipage}%
        \hspace{-0.25cm}
        \begin{minipage}{3.5cm}
            \centering\scriptsize $d=105, \rho=0$
        \end{minipage}
        
        \vspace{-0.2cm} 
        
        \begin{subfigure}[t]{3.5cm}
            \centering
            \includegraphics[width=\linewidth, trim=15 15 15 15, clip]{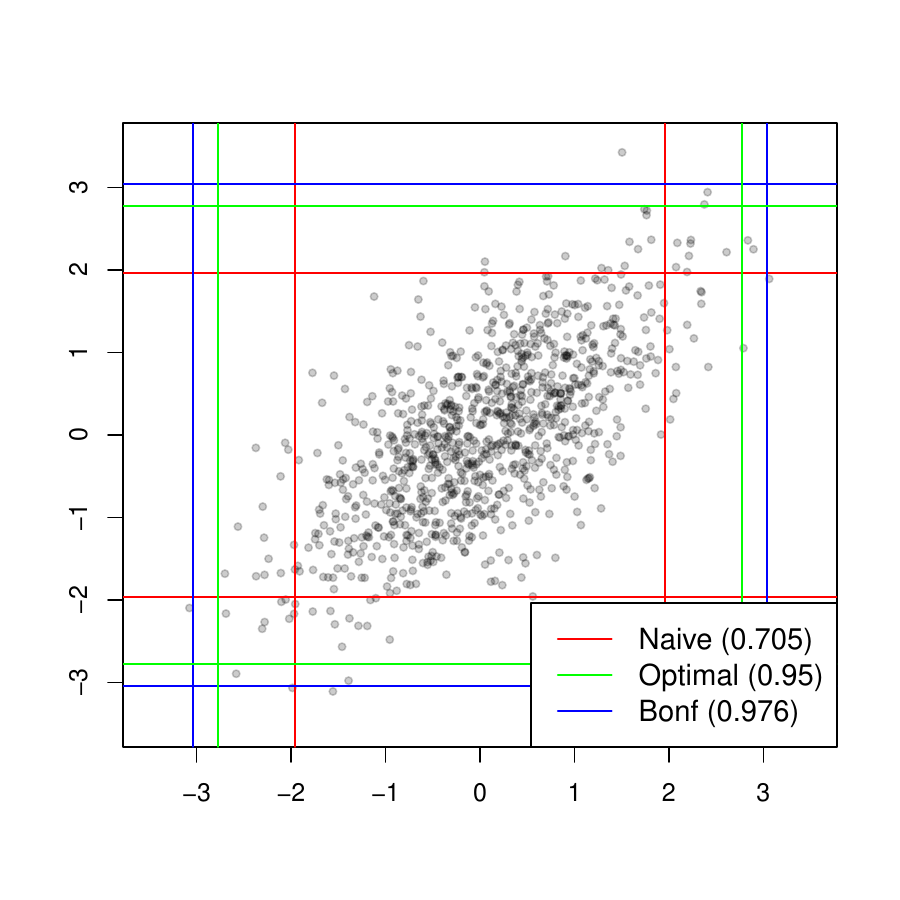}
        \end{subfigure}%
        \hspace{-0.25cm}
        \begin{subfigure}[t]{3.5cm}
            \centering
            \includegraphics[width=\linewidth, trim=15 15 15 15, clip]{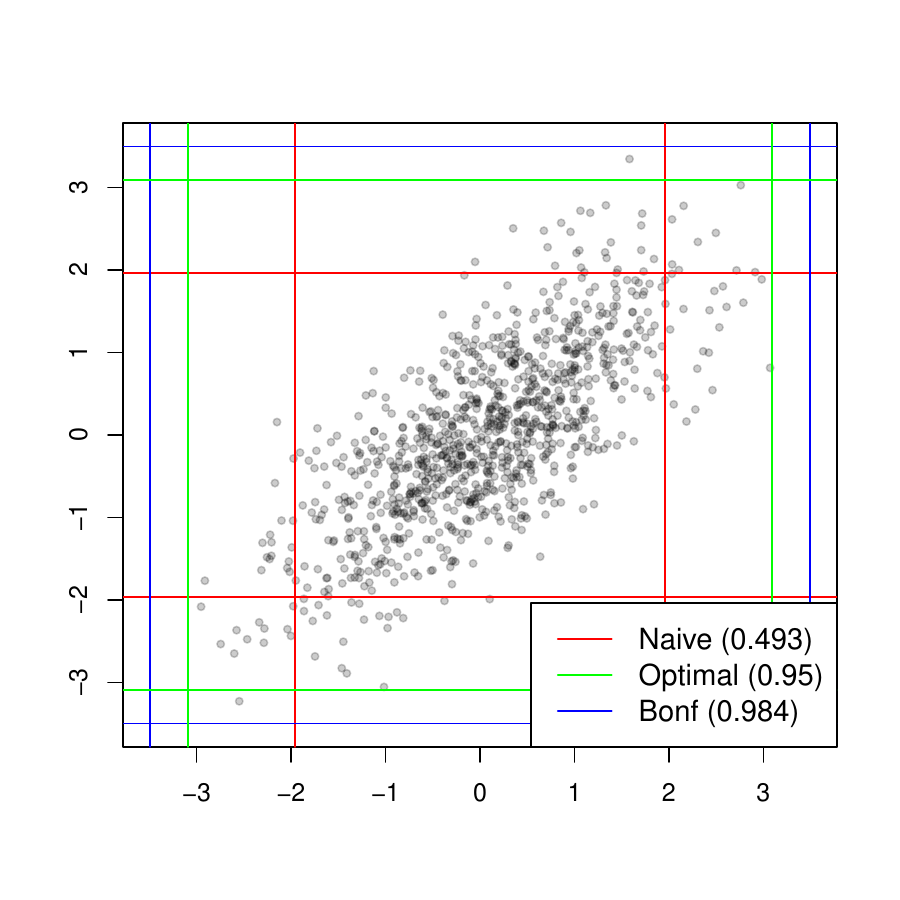}
        \end{subfigure}
        
        \vspace{-0.35cm}
        \begin{minipage}{3.5cm}
            \centering\scriptsize $d=21, \rho=0.7$
        \end{minipage}%
        \hspace{-0.25cm}
        \begin{minipage}{3.5cm}
            \centering\scriptsize $d=105, \rho=0.7$
        \end{minipage}
        
        \caption{$\mathcal{N}_d(0,\mathbf{\Sigma}(\rho))$.}
        \label{fig:normal_all}
    \end{subfigure}%
    ~
    \begin{subfigure}[t]{0.49\textwidth}
        \centering
        
        \vspace{-0.2cm}
        \begin{subfigure}[t]{3.5cm}
            \centering
            \includegraphics[width=\linewidth, trim=15 15 15 15, clip]{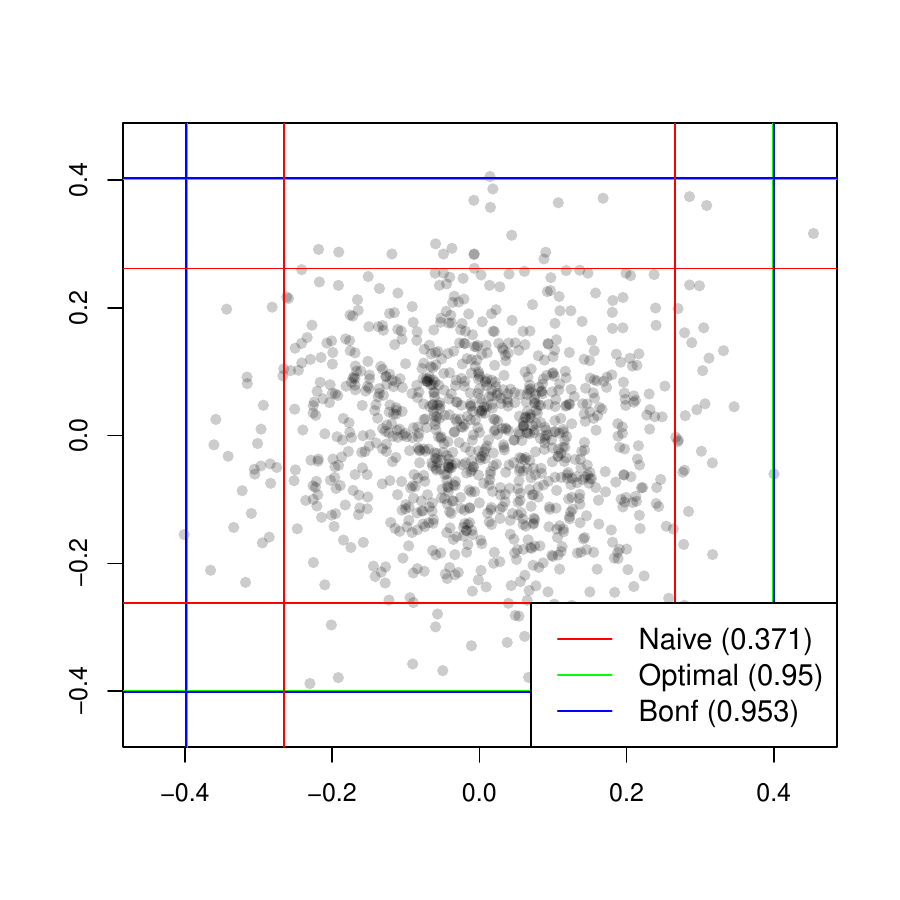}
        \end{subfigure}%
        \hspace{-0.25cm}
        \begin{subfigure}[t]{3.5cm}
            \centering
            \includegraphics[width=\linewidth, trim=15 15 15 15, clip]{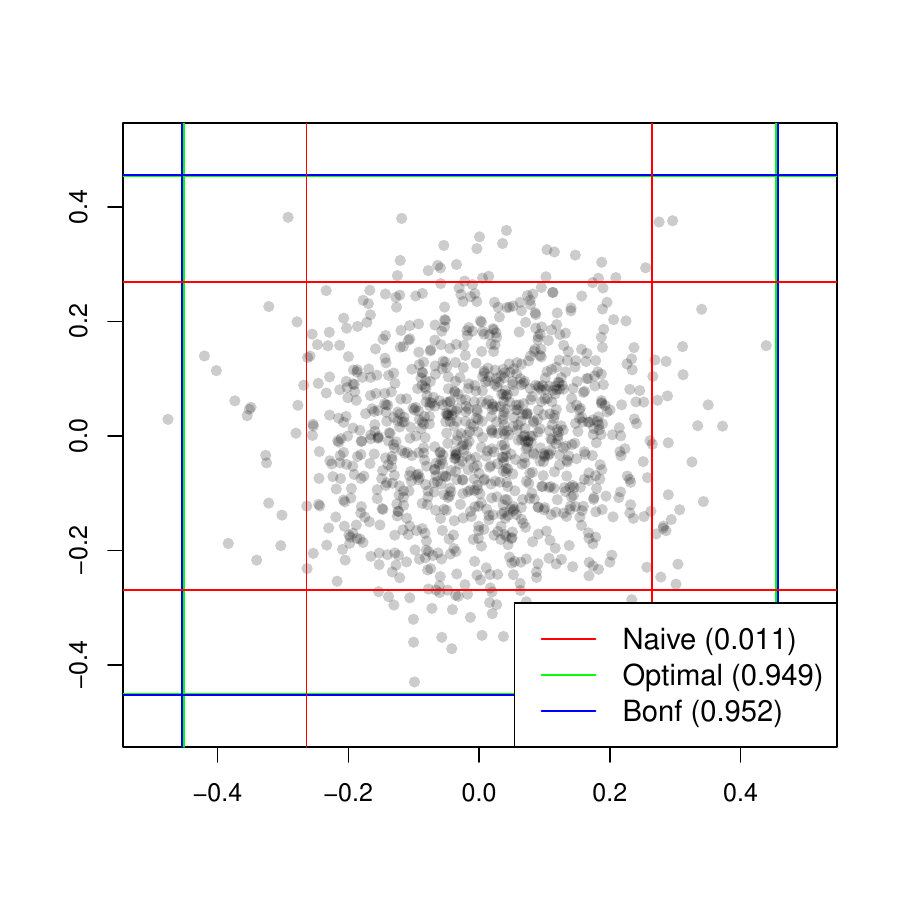}
        \end{subfigure}
        
        \vspace{-0.35cm}
        \begin{minipage}{3.5cm}
            \centering\scriptsize $p=7, \rho=0$
        \end{minipage}%
        \hspace{-0.25cm}
        \begin{minipage}{3.5cm}
            \centering\scriptsize $p=15, \rho=0$
        \end{minipage}
        
        \vspace{-0.2cm}
        
        \begin{subfigure}[t]{3.5cm}
            \centering
            \includegraphics[width=\linewidth, trim=15 15 15 15, clip]{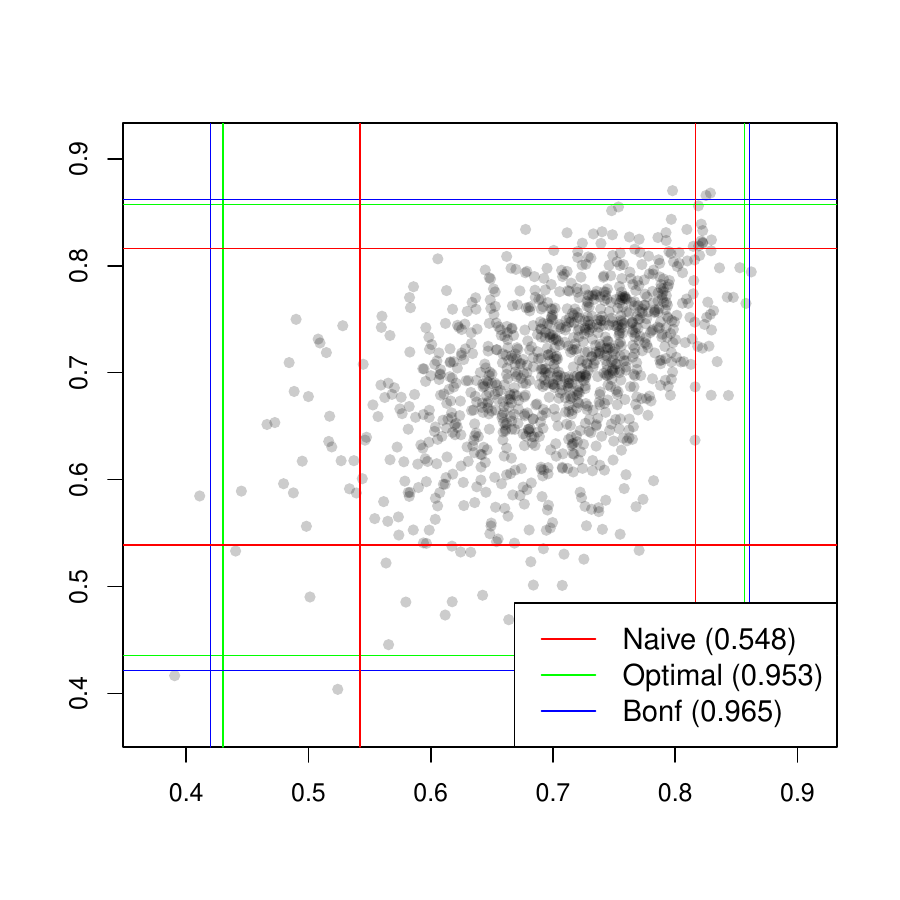}
        \end{subfigure}%
        \hspace{-0.25cm}
        \begin{subfigure}[t]{3.5cm}
            \centering
            \includegraphics[width=\linewidth, trim=15 15 15 15, clip]{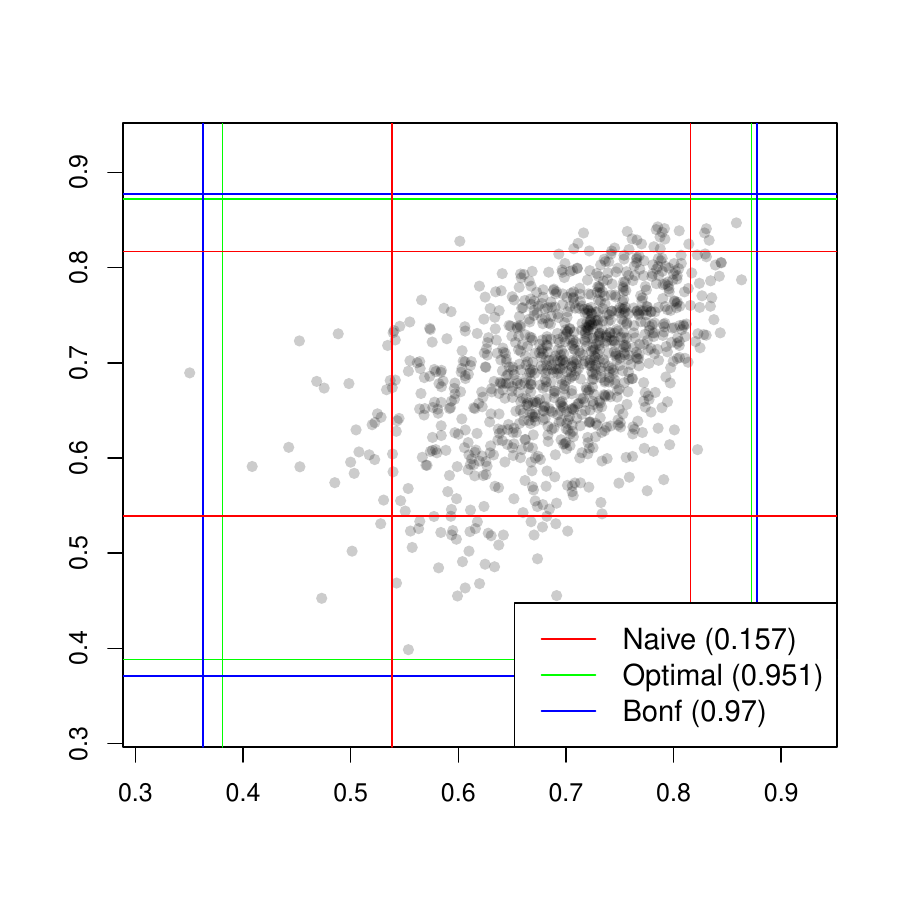}
        \end{subfigure}
        
        \vspace{-0.35cm}
        \begin{minipage}{3.5cm}
            \centering\scriptsize $p=7, \rho=0.7$
        \end{minipage}%
        \hspace{-0.25cm}
        \begin{minipage}{3.5cm}
            \centering\scriptsize $p=15, \rho=0.7$
        \end{minipage}
        
        \caption{$\mathcal{IW}_p(\mathbf{\Sigma}(\rho), \nu=50)$.}
        \label{fig:IW_all}
    \end{subfigure}
    
    \caption{Projection of Naive (red), Bonferroni-type (blue) and Optimal (green) rectangles. Points represent $1000$ samples projected in the first two dimensions. Coverage probabilities are estimated via Monte Carlo ($10^5$ samples). In subfigure \textbf{(a)} Multivariate Normal distribution is used, parameterised like in Equation~\ref{eq:sigmarho}, subfigure \textbf{(b)} use the Inverse-Wishart distribution defined in Equation~\ref{eq:inversewishartexemple}. Both are simulated for different dimension and different values of $\rho$ that controls the dependencies.}
    \label{fig:rectangles_comparison}
\end{figure*}

These toy examples help us to understand the effect of dimension and the dependencies among the vector $\theta$ on the rectangles. The two types of distribution are chosen to illustrate Gaussian and non-Gaussian behaviour. The hyperparameters $d$ and $p$ are chosen so that the dimension of the parameter of interest is the same overall. Finally, the value of the hyperparameter $\rho$ for both should affect the dependencies similarly. After computing the different rectangles, we project them, along with the samples, onto the first two dimensions in Figure~\ref{fig:rectangles_comparison}. 

Figure~\ref{fig:normal_all} shows that, in the case of independence where $\rho$ is equal to zero, the naive rectangle has a probability of $(1 - \alpha)^d$. For example, we have $36\%$ coverage in dimension 20. As $d$ increases, this illustrates how the naive rectangle can lead to catastrophic results, with none of the samples inside it. Conversely, the Bonferroni example is optimal in this setup of independence, with a coverage of nearly $95\%$ in these cases. However, as the correlation increases, the Bonferroni rectangle becomes overly conservative. This resonates with the criticism of the Bonferroni procedure in multiple testing, where it is considered too conservative if the variables are dependent.

In our application, we work with correlation matrices with an inverse-Wishart distribution on the covariance. This second toy distribution is more suitable for demonstrating what can happen in our study. Figure~\ref{fig:IW_all} shows that the same results are observed in the coverage of the rectangles. However, it is important to note that the rectangles are no longer symmetric, as the correlation transformation makes the distribution skewed. It should be noted that a higher value for the scale matrix parameter in the inverse-Wishart distribution induces higher dependencies between the coefficients; hence the need to use the optimal rectangle. In real datasets, these coefficients are around 0.4, which induces strong dependencies between the coefficients. Therefore, working with the Bonferroni rectangle is inadequate and is not used in the experiments on real datasets in Section~\ref{sec:realdatasets}.

\subsubsection{Navigating computational complexity} \label{sec:computation}

This section focuses on estimating the optimal rectangle, as the naive rectangle fails to constitute a valid credible region and the Bonferroni rectangle often lacks precision. The primary computational bottleneck arises from estimating extreme quantiles. Since the optimal level $t_{\alpha}$ lies within $[\frac{\alpha}{d}, \alpha]$, accurate estimation may require targeting levels as low as $\frac{\alpha}{2d}$. Achieving this precision typically demands $\mathcal{O}(\frac{d}{\alpha})$ Monte Carlo samples. Storing the full $d$-dimensional vector for each sample results in a memory complexity of $\mathcal{O}(d^2)$, which quickly becomes infeasible as $d$ increases. Here, we present algorithms designed to estimate the optimal rectangle while mitigating these memory constraints. In Section~\ref{sec:realdatasets} where studying coma datasets with $p=107$ ROIs, the correlation matrix can be flatten as a vector with $d = \frac{107 \times 106}{2}$ dimensions. For a level of credibility $\alpha=0.05$, the target quantile level can drop to approximately $4.10^{-6}$. Estimating such extreme quantiles is notoriously unstable; with fewer than $10^6$ realizations, empirical estimators collapse to the sample minimum or maximum, providing little statistical information. This instability necessitates specialized handling, as detailed in the following algorithms.

\paragraph*{BGHM \citep[Besag, Green, Higdon, Mengersen,][method for low dimensions]{besag_bayesian_1995}}

A method to infer an empirical credible region is proposed in \cite{besag_bayesian_1995}, it has been popularised in \cite{held_simultaneous_2004} and implemented in the package \textit{BayesSurv} \citep{komarek_package_2015}. It estimates a credible region for $\theta \in \mathbb{R}^d$ at a level $1-\alpha$ using marginal order statistics based on samples. This method is of particular interest as it provides an estimation of both the region $R(t_{\alpha})$ and the threshold $t_{\alpha}$ directly. This spares us from the two-step exploratory method that relies on first estimating  $\widehat{R}(t)$ for several values of $t$, and then estimating the probability of such region. The algorithm is described in Algorithm~\ref{algo:besag}.

\begin{algorithm}
\caption{BGHM's rectangle (low dimensions)}
\label{algo:besag}
\KwIn{$M$ samples $\theta_1,\dots,\theta_M$, $k$ an integer lower than $M$}
\KwOut{A hyperrectangle covering $\ge k$ samples}
  \BlankLine
  \For{each dimension $i$}{
    $r_i^{(j)} \gets \text{rank of sample } \theta_j \text{ along dimension $i$}$ \;
  }
  \BlankLine
\For{each sample $j$}{
    $S^{(j)} \gets \max\bigl(M + 1 - \min_i r_i^{(j)},\; \max_i r_i^{(j)}\bigr)$
}
  \BlankLine
$j^* \gets k\text{-th smallest } S^{(j)}$
  \BlankLine
\For{each dimension $i$}{
    $L_i \gets (M+1-j^*)\text{-th smallest value in dim } i$\;
    $U_i \gets j^*\text{-th smallest value in dim } i$\;
}
  \BlankLine
\Return $R=\prod_i [L_i, U_i]$
\end{algorithm}

\begin{proposition}
\label{prop:bghmoptimal}
The rectangle constructed by the BGHM algorithm in the case where $k$ is equal to $\lceil (1-\alpha) M \rceil$ provides an estimation of the optimal rectangle $R(t_{\alpha})$:
\[
\prod_{i=1}^{d}\bigl[\theta_i^{[M+1-j^{\star}]},\ \theta_i^{[j^{\star}]}\bigr]=\prod_{i=1}^{d}\bigl[\,\widehat{q}_i\left(\frac{t_{\alpha}}{2}\right),\;\widehat{q}_i\left(1-\frac{t_{\alpha}}{2}\right)\bigr].
\]
It also provides an estimation of $t_{\alpha}$: 
\[\hat{t}_{\alpha}=\frac{2(M+1-j^{\star})}{M+1}.\]
\end{proposition}

This estimator requires storing all samples in memory. However, as previously established, the number of samples required for an accurate estimation increases quadratically with the dimension. For low dimensions (d < 20), the implementation of this algorithm (e.g., \cite{held_simultaneous_2004,komarek_package_2015}) works extremely well, but it becomes infeasible in higher dimensions. For moderate dimensions ($d<200$), we have developed an online version of the previous algorithm.

\paragraph*{Online BGHM adaptation (method for moderate dimensions)}
The BGHM rectangle is based on identifying the $\alpha M$  points with the largest values of the function $S$ applied to the entire set of samples $M$. In this online version, we propose identifying these extreme points in batches, which allows us to keep only the necessary number of points in memory. However, since $S$ is based on ranks, it cannot be determined for each point separately.

To obtain an algorithm that determines just the k points with the largest $S$ values without keeping all the samples in memory, it is possible to work with a new function $S(.,\mathcal{D})$ defined on every set $\mathcal{D}$. For each point $\theta^{(j)}$ in a set $\mathcal{D}$, define its score as

\[
S(\theta^{(j)}, \mathcal{D}) = \max\left( |\mathcal{D}| + 1 - \min_{i=1,\dots,d} r_i^{(j)}(\mathcal{D}),\; \max_{i=1,\dots,d} r_i^{(j)}(\mathcal{D}) \right),
\]
where $r_i^{(j)}(\mathcal{D})$ denotes the rank of the $i$-th coordinate of $\theta^{(j)}$ among the points in $\mathcal{D}$.

\begin{lemma}[Extreme points preservation]
\label{prop:extremepoints}
Let $A, B \subset \mathbb{R}^d$ be two disjoint batches of points.  
For any integer $m \ge 1$, define:
\[
\mathcal{E}_m(A) = \text{the set of $m$ points in $A$ with the largest values of } S(\cdot, A),
\]
and similarly $\mathcal{E}_m(B)$ for $B$ and $\mathcal{E}_m(A \cup B)$ for $A \cup B$.

Then:
\[
\mathcal{E}_m(A \cup B) \subseteq \mathcal{E}_m(A) \cup \mathcal{E}_m(B).
\]

In other words, the global extreme points of the union are contained in the union of the batch-wise extremes. 
\end{lemma}

\begin{algorithm}[h!]
  \caption{Online BGHM Algorithm (moderate dimensions)}
  \label{algo:online_besag}
  \KwIn{Stream of batches $\mathcal{B}_1, \mathcal{B}_2, \dots$, coverage target $1-\alpha$, total size $M$}
  \KwOut{Rectangle $R^*$ covering exactly $\lfloor (1-\alpha)M \rfloor$ samples of $\theta$}
  \BlankLine
  $m \gets \lceil \alpha M \rceil$ \tcp*{lower bound of the number of extreme points to keep}
  $\mathcal{E} \gets \emptyset$ 
  $M_{\text{total}} \gets 0$ \;
  \BlankLine
  \For{each batch $\mathcal{B}$}{
    $M_{\text{total}} \gets M_{\text{total}} + |\mathcal{B}|$ \;
    $\mathcal{E} \gets \text{at least m elements of } \mathcal{E}\cup \mathcal{B} \text{ with largest values of} S(.,\mathcal{E}\cup \mathcal{B})$ 
    \tcp*{due to possible equalities of $S$ it may contain more than $m$ elements}
  }
  \BlankLine
  $M \gets M_{\text{total}}$ \;
  $m' \gets |\mathcal{E}|$  \tcp*{size of the extreme subset}
  \BlankLine

  $t \gets m' - \lfloor \alpha M \rfloor$ \tcp*{exact number of extreme points to retain}
  \BlankLine
  \For{$i = 1$ \KwTo $d$}{
    $\theta_i^{\min} \gets t\text{-th smallest value of coordinate } i \text{ in } \mathcal{E}$ \;
    $\theta_i^{\max} \gets t\text{-th largest value of coordinate } i \text{ in } \mathcal{E}$ \;
  }
  \BlankLine
  $R^* \gets [\theta_1^{\min}, \theta_1^{\max}] \times \dots \times [\theta_d^{\min}, \theta_d^{\max}]$ \;
  \BlankLine
  \Return $R^*$
\end{algorithm}

Consequently, evaluating $S$ on the reduced set, $\mathcal{E}_m(A) \cup \mathcal{E}_m(B)$, allows us to recover $\mathcal{E}_m(A \cup B)$ without accessing all the points of $A \cup B$. This spares us the need to evaluate $S(., A \cup B)$, which is the step where we had to remember all the points of $A \cup B$.

Using the previous lemma, we can compute the set of extreme points, denoted by $\mathcal{E}$, iteratively. Since this set contains only the $m$ most extreme points, if $m > \lfloor\alpha M\rfloor$, we can identify the $m - \lfloor\alpha M\rfloor$-th order statistic of $S(.,\mathcal{E})$, which corresponds to the $M - \lfloor\alpha M\rfloor$-th order statistic of the set of $S(\theta_i), i = 1, ..., M$, which is used to construct the BGHM rectangle in Algorithm~\ref{algo:besag}. The only remaining step is to compute the corresponding rectangle. This reduces the memory requirement to just $\alpha M$ samples of $\theta$. This process is detailed is described in Algorithm~\ref{algo:online_besag}. The online BGHM does not reduce the asymptotic memory complexity compared to the BGHM algorithm, as it only divides the memory by a constant factor of  $\frac{1}{\alpha}$. However, this online version enables us to estimate the BGHM rectangles when $d\approx100$ on a regular computer.

\paragraph*{Sliced-quantile grid algorithm (method for high dimensions)}
To reduce memory costs when handling $d\approx1000$, we propose a two-stage exploratory approach that is described in Algorithm~\ref{algo:sliced_quantile}.

\textit{Stage 1 – Sliced quantile estimation:}  Based on a grid of values $t_1, ..., t_K$, where $t_1 = \alpha/d < t_2 < ... < t_K = \alpha$, we first estimate the quantiles for each dimension separately: $q_i(t_1), ..., q_i(t_K)$. This means that, for each marginal $i$, we need to sample many points, but we only need to keep the $M$ samples of this coefficient in memory, not the d-dimensional vector. We can then process these estimation by batches using order statistics. In this case the maximum storage requirement for each batch is $\text{max}(t_k) M=\alpha M$, where $M$ is the number of samples. At the end of this process, we have an estimation for each region, $R(t_1),...,R(t_K)$.

\textit{Stage 2 – Coverage‑based selection:} Once we have the candidate rectangles, we evaluate each one by estimating its actual coverage via a second Monte Carlo sample from the posterior. The smallest rectangle achieving at least $1-\alpha$ coverage is selected. This accounts for the multivariate dependencies ignored in the first stage. This step also stabilises the result because coverage estimation is statistically more reliable than extreme quantile estimation. It serves as a validation step and ensures that we select a rectangle with the correct property, even if the quantiles are very unstable.

This method sacrifices some precision in the estimated quantile level $t_{\alpha}$ due to the grid required. However, the size of this one-dimensional grid for $t$ is not limited, as adding a value to it has a negligible computational impact. In practice the only limitation is to take $t_{\alpha}^{\text{bonf}}$ as a lower bound and $t_{\alpha}^{\text{naive}}$ as a upper bound for the grid, thanks to the results from Proposition~\ref{prop:specialrectangles}. 

\begin{algorithm}
  \caption{Sliced-quantile grid algorithm (high dimensions)}
  \label{algo:sliced_quantile}
  \KwIn{Batches $\mathcal{B}_1,\dots$, target $1-\alpha$, grid $t_1,\dots,t_K\in[\frac{\alpha}{d},\alpha]$, memory $m=\lceil\max(t_k)M\rceil$}
  \KwOut{$\widehat{R}(t_{\alpha})$ with coverage $\geq 1-\alpha$}
  \BlankLine
  \textbf{Stage 1 – Sliced quantiles estimation} \;
  \For{$j=1..d$}{
  $N\gets0$\;
  $\mathcal{L}_j,\mathcal{U}_j\gets\emptyset$\;
  
    \For{each batch $\mathcal{B}$}{
        $N\gets N+|\mathcal{B}|$ \;
        $\mathcal{L}_j\gets m/2 \text{ lower values in } \mathcal{L}_j\cup\mathcal{B}^{(j)} $ \;
      $\mathcal{U}_j\gets m/2 \text{ upper values in }\mathcal{U}_j\cup\mathcal{B}^{(j)}$ \;
    }
    \For{$k=1..K$}{
    \tcp*{quantile estimation using order statistic by batches}
    $\hat{q}_j(t_k)\gets N t_k/2 \text{-th largest value in } \mathcal{L}_j$ 
    $\hat{q}_j(1-t_k/2)\gets  m-Nt_k/2 \text{-th largest value in } \mathcal{U}_j$\; 
    }
  }
  
  \For{$k=1..K$}{
    $\widehat{R}(t_k)\gets\prod_j[\hat{q}_j(t_k/2),\hat{q}_j(1-t_k/2)]$ \;
    }
  \BlankLine
  \textbf{Stage 2 – Coverage selection} \;
  Draw validation sample $\mathcal{D}_{\text{val}}$ \;
  \For{$k=1..K$}{
    $\widehat{\text{Cov}}_k\gets\frac{1}{|\mathcal{D}_{\text{val}}|}\sum_{\theta\in\mathcal{D}_{\text{val}}} \mathds{1}[\theta\in\widehat{R}(t_k)]$ \;
  }
  $k^\star\gets\min\{k:\widehat{\text{Cov}}_k\geq 1-\alpha\}$ \;
  \Return $\widehat{R}(t_{k^\star})$ \;
\end{algorithm}

In the loop on $j \in [1,d]$ we keep in memory $\alpha M$ samples of each coefficient $j$. This is the same number of samples as before; however, we do not store the d-dimensional vector in its entirety, only the coefficient. This reduces the asymptotic memory complexity from $\mathcal{O}(d^2)$ to $\mathcal{O}(d)$, enabling us to scale well even if $d > 1000$. However, this comes at a higher time cost: indeed, the rectangle is estimated using $M$ samples, but the loop forces us to simulate these $M$ samples $d$ times.

\paragraph*{Recommendations}
\begin{table}[h]

\centering
\caption{Comparison of the 3 algorithms complexity in memory and simulation. The constant in the complexity $\mathcal{O}$ if the number of the sample $M=\mathcal{O}(d)$ is the same for each algorithm. The column feasible $d$ establishes the value of $d$ such that a reasonable estimation of the rectangle can be estimated on a regular computer.}
\label{tab:algos}
\begin{tabular}{lcccc}
\toprule
{Algorithm} & {Maximum memory used} & {Number of simulations} & {Feasible $d$}\\
\midrule
{BGHM \citep{besag_bayesian_1995}} & $\mathcal{O}(d^2)$ & $\mathcal{O}(d)$ & $d\approx 10$ \\
{Online BGHM (ours)} & $\mathcal{O}(\alpha d^2)$ & $\mathcal{O}(d)$ & $d\approx 10^2$ \\
{Sliced quantile grid algorithm (ours)} & $\mathcal{O}(d)$ & $\mathcal{O}(d^2)$ & $d \approx 10^4$\\
\bottomrule
\end{tabular}

\end{table}

Table~\ref{tab:algos} summarises the memory and sample costs. This table represents the asymptotic complexity; however, it does not clearly show the difference between the BGHM algorithm and our online adaptation of this algorithm. Overall, the online adaptation has little impact on time complexity and enables us to reduce the memory cost by $\alpha$. Even if this is just a constant, this is a meaningful improvement, as it can be used for dimensions around 100, which could represent, for instance, the nodes of the brain in neuroscience. For the experiments in Section~\ref{sec:realdatasets}, the algorithm used is Sliced Quantile Estimation due to the dimension of $5000$. The introduction of Online BGHM provides tools for modelling lighter datasets.

\section{Specificities for correlation matrices} \label{sec:correlation}

This section focuses on the case of fMRI study where we need to model the connectivity between brain regions. In this set-up we have a $p$-dimensionnal vector $X=(X_1,\dots,X_p)$ that represents the signal coming from each of $p$ brain Regions of Interest (ROIs). The parameter of interest is the correlation matrix $\boldsymbol{C}$ between these variables, where for each pair $(i,j)$, we have $C_{ij}=\text{Corr}(X_i,X_j)$.

\subsection{Model and guarantees} \label{sec:model}

\paragraph*{Model and prior choice}
Data from brain ROIs are modelled using a multivariate zero mean Gaussian distribution. The parameter to estimate is the covariance matrix, denoted by $\boldsymbol{\Sigma}$. We choose to work with an inverse-Wishart prior to model the covariance matrix, as this is conjugate with the multivariate Gaussian distribution. The model is written as follows:

\[
X_i\mid \boldsymbol{\Sigma} \overset{i.i.d.}{\sim} \mathcal{N}_p\left(\mathbf{0}, \mathbf{\boldsymbol{\boldsymbol{\Sigma}}}\right), \quad \boldsymbol{\boldsymbol{\Sigma}} \sim \mathcal{IW}(\Phi,\nu).
\]

For fMRI analysis, our parameter of interest is the correlation matrix $\boldsymbol{C}$. Fortunately, determining the posterior distribution of the covariance matrix $\boldsymbol{\Sigma}$ is equivalent to determining the posterior distribution of the correlation matrix $\boldsymbol{C}$, since it is fully determined by $\boldsymbol{\Sigma}$ thanks to the following transformation:

\begin{equation}
    \boldsymbol{C} = \boldsymbol{D}^{-\frac{1}{2}} \boldsymbol{\boldsymbol{\Sigma}} \boldsymbol{D}^{-\frac{1}{2}} \quad \text{with} \quad \boldsymbol{D}=\mathrm{diag}(\boldsymbol{\boldsymbol{\Sigma}}).
    \label{eq:cor_transform}
\end{equation}

We have $n$ realisations of $X$ that is denoted $\mathbf{X}_n \in \mathbb{R}^{n \times p}$.  The conjugate nature of the inverse-Wishart distribution with the Gaussian multivariate \citep{haff_identity_1979} give us a closed-form expression for the posterior distribution of $\boldsymbol{\Sigma}$: \[ \boldsymbol{\boldsymbol{\Sigma}} |\mathbf{X}_n \sim \mathcal{IW}(\Phi+\mathbf{X}_n^T \mathbf{X}_n,\nu+n).\]  

Using the transformation in Equation~\ref{eq:cor_transform} we have access to the posterior distribution of $\boldsymbol{C}$.
The hyperparameters $\Phi$ and $\nu$ can be chosen according to expert knowledge, if available. However, in fMRI connectivity analysis, prior information about the covariance structure of the data is usually unavailable. For this reason, we choose to work with a low-information prior that slightly shrinks the posterior distribution towards the identity matrix. This prior favours independence between brain regions, which is a reasonable assumption when no information is available about the connectivity structure of the brain. To achieve this, we select: 
\[\nu=p+2, \quad \Phi=\mathrm{diag}(\hat{\sigma}_1,\dots,\hat{\sigma}_p), \quad \text{ which ensures } \quad \mathbb{E}[\mathbf{\boldsymbol{C}}]=I.\]

Choice of hyperparameters and their impact on prior and posterior moments are discussed in \cite{jiang_prior_2026}.

\paragraph*{Bernstein--von Mises Theorem for correlation matrices}

Under a Gaussian likelihood with zero mean, when the prior on the covariance matrix $\boldsymbol{\boldsymbol{\Sigma}}$ belongs to a class satisfying appropriate regularity conditions, and the true covariance $\boldsymbol{\boldsymbol{\Sigma}}_0$ is positive definite, Theorem~4.4 of \citet{sarkar_high-dimensional_2025} establishes a Bernstein--von Mises theorem for covariance matrices. This theorem guarantees the asymptotic normality of the posterior distribution of $\boldsymbol{\boldsymbol{\Sigma}}$ under their Assumptions~1--3. These assumptions are verified, for instance, in the case of an inverse-Wishart prior when the eigenvalues of $\boldsymbol{\boldsymbol{\Sigma}}_0$ are bounded below by a strictly positive constant. The theorem can be reformulated in terms of the vector of unique elements of a symmetric matrix. Define the half-vectorization operators $\textup{vech} : \mathbb{R}^{p \times p} \to \mathbb{R}^{p(p+1)/2}$ (resp. $\textup{vech}_0 : \mathbb{R}^{p \times p} \to \mathbb{R}^{p(p-1)/2}$) which stack the columns of the upper triangular part (resp. strict upper triangular part) of its argument. Then we have:
\[
\sqrt{n}\,\textup{vech}(\boldsymbol{\boldsymbol{\Sigma}} - \boldsymbol{S}_n) | \mathbf{X}_n 
\;\stackrel{\mathrm{TV}}{\approx}\; 
\mathrm{N}_{p(p+1)/2}\bigl( \mathbf{0},\, 2 \boldsymbol{B}_p^\top (\boldsymbol{S}_n \otimes \boldsymbol{S}_n) \boldsymbol{B}_p \bigr),
\]
where the approximation is in the sense that the total variation distance between the two distributions converges to zero in probability as $n \to \infty$, and where $\boldsymbol{S}_n = \frac{1}{n}\sum_{i=1}^n \boldsymbol{X}_i\boldsymbol{X}_i^\top$ is the empirical covariance matrix (assuming zero mean), which is a consistent estimator of $\boldsymbol{\boldsymbol{\Sigma}}_0$, and $\boldsymbol{B}_p$ denotes the duplication matrix (whose explicit form is given in \cite{sarkar_high-dimensional_2025}).

In our setting, our primary interest lies in the posterior distribution of the correlation matrix $\boldsymbol{C}$. Thus the following theorem establishes similar results for correlation matrices, proof can be seen in the appendix.

\begin{theorem}[Bernstein--von Mises theorem for correlation matrices]
\label{theorem:bmv}

Under a Gaussian likelihood with zero mean, when the prior on the covariance matrix is an inverse-Wishart distribution and if the eigenvalues of $\boldsymbol{\boldsymbol{\Sigma}}_0$ are bounded below by a strictly positive constant, then the strict upper-triangular vectorization $\textup{vech}_0 : \mathbb{R}^{p \times p} \to \mathbb{R}^{p(p-1)/2}$ of the corresponding correlation matrix $\boldsymbol{C}$ has the following property:
\[
\sqrt{n}\,\textup{vech}_0(\boldsymbol{C} - \boldsymbol{C}_n) | \mathbf{X}_n 
\;\stackrel{\mathrm{TV}}{\approx}\; 
\mathrm{N}_{p(p-1)/2}\bigl( \mathbf{0}, \boldsymbol{V} \bigr),
\]

where $\boldsymbol{C}_n$ is the empirical correlation matrix derived from $\boldsymbol{S}_n$, for some PSD matrix $\boldsymbol{V}$.
\end{theorem}

In the context of this paper, we use the Bernstein--von Mises result to demonstrate that our credible regions become increasingly accurate approximations of frequentist confidence regions as the sample size increases. However, we do not rely on the asymptotic distribution to construct the regions, and prefer to use an approach that is valid for any sample size. The asymptotic Gaussian behaviour also motivates using a rectangular region as a multivariate Gaussian and its marginals are all unimodal.

At first glance, one might worry that the vectorization loses information, as many vectors in $\mathbb{R}^{p(p-1)/2}$ do not correspond to a valid positive-definite correlation matrix when applying the inverse of the $\textup{vech}_0$ function. However, the Bernstein--von Mises theorem guarantees that, asymptotically, the posterior distribution of $\sqrt{n}\,\textup{vech}_0(\boldsymbol{C} - \widehat{\boldsymbol{C}}_n)$ is Gaussian and concentrated on a $\sqrt{n}$-neighborhood of the true value. In this local regime, the positive-definiteness constraint is asymptotically negligible, and the mapping between the vector of off-diagonal elements and the full correlation matrix is bijective onto its image.

\subsection{Estimation of the support of the matrices}

In fMRI connectivity analysis, identifying significant correlations between brain regions is crucial for understanding functional connectivity. We propose to use the previously defined credible regions to detect non-zero correlations, which can be seen as significant edges of the connectivity graph. Let $\boldsymbol{C}_0$ be the true $p \times p$ correlation matrix, and denote its support by 
\begin{equation*}
    A_0 = \{(i,j),i<j: \boldsymbol{C}_{0,ij} \neq 0\}.
\end{equation*}

$A_0$ represents the set of index pairs corresponding to non-zero correlations, which can be interpreted as the adjacency matrix of the true underlying graph.

Given $n$ observations, let $\pi(\textup{vech}_0(\boldsymbol{C}) | \mathbf{X}_n)$ be the posterior distribution of the upper coefficients of $\boldsymbol{C}$, 
and denote by $R(t_{\alpha_n})$ a credible rectangle at a level $1-\alpha_n$ for $\textup{vech}_0(\boldsymbol{C})$. This rectangle is in $\mathbb{R}^{\frac{p(p-1)}{2}}$, however to understand the link with the matrix it can be rewritten as: 
\[R(t_{\alpha_n})=\prod_{i<j} I_{ij}(t_{\alpha_n}), \quad \text{ with } I_{ij}(t_{\alpha_n})=\left[q_{ij}\left(\frac{t_{\alpha_n}}{2}\right),q_{ij}\left(1-\frac{t_{\alpha_n}}{2}\right)\right].\]

For this section we consider the support estimator:
\[
\widehat{A}_n = \left\{(i,j): 0 \notin I_{ij}(t_{\alpha_n}) \right\}, \quad \text{where $\alpha_n \to 0$ as $n \to \infty$}.
\]

\paragraph*{Consistency of the support estimator}
We ultimately establish that for a properly decaying sequence $\alpha_n$, the credible intervals achieve to reliably discriminate zero coefficients from non-zero coefficients to ensure exact support recovery. This establishes the consistency of the support estimator.

\begin{proposition}
\label{prop:consistency}
$\widehat{A}_n$ converges in probability to $A_0$ as $n \to \infty$, i.e., the support estimator is consistent and for any sequence $\alpha_n \to 0$ that respects the Gaussian tail decay implied by the BvM, the estimated support converges to the true support in probability:
\[
\mathbb{P}\bigl(\widehat{A}_n = A_0\bigr) \xrightarrow[n\to\infty]{} 1.
\]
\end{proposition}

The requirement $\alpha_n \to 0$ ensures that the credible intervals cover the true parameter $\boldsymbol{C}_0$ with probability tending to one. However, for non-zero entries ($\boldsymbol{C}_{0,ij} \neq 0$), the interval must exclude zero asymptotically. This imposes a balance: while decreasing $\alpha_n$ widens the interval (via larger quantiles), this expansion must be dominated by the $1/\sqrt{n}$ contraction rate driven by the sample size. The Bernstein--von Mises theorem guarantees the necessary Gaussian tail behavior, ensuring that vanishing $\alpha_n$ eliminates false discoveries (by shrinking intervals around zero) while preserving true signals (by maintaining exclusion of zero for non-zero coefficients).


\paragraph*{Finite-sample properties: connection to frequentist multiple testing}
In practice we work with a fixed sample size $n$ and a fixed level for the credible region $1-\alpha$. The estimator $\widehat{A}_n$ still verifies some properties for $\alpha_n=\alpha$ fixed thanks to the construction of the credible region. Indeed the posterior probability of the region is valid for a finite-sample as it is not constructed using asymptotic results. Futhermore the support estimator is equivalent to the decision procedure defined in Section~\ref{sec:comparaisonreference} in the case where the reference value is the vector of zeros. This decision $\delta$ verifies:

\[
P\left( \bigcup_{i=1}^m \{ \delta_{i,j}=1, C_{i,j} = 0 \} \;\Big|\; \mathbf{X}_n \right) \leq \alpha.
\]

To illustrate what this property represents, it can be rewritten as the error rate of a multiple testing procedure. As mentioned in the introduction, the support of the correlation matrix is usually estimated by performing a hypothesis test for each coefficient $\boldsymbol{C}_{ij}$ with the null hypothesis $H_{0,ij}: \boldsymbol{C}_{ij}=0$ against the alternative $H_{1,ij}: \boldsymbol{C}_{ij} \neq 0$. One correlation test provides control over the probability of one false positive occurring, whereby the null hypothesis has been rejected incorrectly. A multiple testing procedure then provides tools to extend this control from one individual test to an error rate across a set of tests. An interesting type of error at the global level is the family-wise error rate (FWER). The FWER is a specific error control defined under the complete null hypothesis $H_0^C$, which assumes that all null hypotheses are true simultaneously. FWER is then defined as the probability of obtaining at least one false positive result under this complete null hypothesis:
\[
\text{FWER} = \mathbb{P}(\text{there is at least one false positive} | H_0^C).
\]

The usual procedure for controlling the FWER is the Bonferroni correction \citep{bonferroni_teoria_1936}, which does not work with joint distributions under $H_0^C$, but rather with marginal distributions under $H_{0,ij}$. Error control is thus strict and conservative, especially if there are many dependencies. This procedure has been heavily criticised in fMRI studies, as many believe that it does not detect anything \citep{epskamp_estimating_2018}. Other methods are known, such as the Holm--Bonferroni method \citep{holm_simple_1979}, that is slighty less conservative.

In our case we can interpret our estimator $\widehat{A}_n$ as a control on the probability of having a false positive but conditionally on the data. In a continuous Bayesian framework, the notion of a ``false positive'' (detecting a non-existent difference) cannot be defined simply by the equality $\theta_i = \theta_{\text{ref},i}$, as this event has posterior probability zero. Instead, we can interpret the false positive as an erroneous directional claim.

\begin{definition}[Bayesian False Positive]

Let the decision for component $i$ be decomposed into two directional claims: $\delta_i^+ = 1$ (claiming a positive shift, $\theta_i > \theta_{\text{ref},i}$) and $\delta_i^- = 1$ (claiming a negative shift, $\theta_i < \theta_{\text{ref},i}$). A Bayesian False Positive occurs on component $i$ if a directional claim is made that contradicts the true parameter value. Formally:
\begin{itemize}
    \item A \textit{false positive increase} occurs if $\delta_i^+ = 1$ and $\theta_i \leq \theta_{\text{ref},i}$.
    \item A \textit{false positive decrease} occurs if $\delta_i^- = 1$ and $\theta_i \geq \theta_{\text{ref},i}$.
\end{itemize}
The event of a Bayesian false positive on component $i$, denoted by $\text{FP}_i$, is defined as the union of these two mutually exclusive errors:
\[
\text{FP}_i = \left( \{ \delta_i^+ = 1 \} \cap \{ \theta_i \leq \theta_{\text{ref},i} \} \right) \cup \left( \{ \delta_i^- = 1 \} \cap \{ \theta_i \geq \theta_{\text{ref},i} \} \right).
\]
\end{definition}

This allows us to define a Bayesian version of the FWER.

\begin{definition}[Bayesian Family-Wise Error Rate (BFWER)]
\label{def:bfwer}
The Bayesian Family-Wise Error Rate (BFWER) is defined as the posterior probability of making at least one Bayesian false positive across all $d$ components. It quantifies the risk of incorrectly inferring either the existence or the direction of a difference in any part of the parameter vector.

Formally, for a decision procedure characterized by the vectors $\delta^+$ and $\delta^-$, the BFWER is:
\begin{align*}
    \text{BFWER} &= \mathbb{P}\left( \bigcup_{i=1}^d \text{FP}_i \;\Bigg|\; \mathbf{y} \right) \\
    &= \mathbb{P}\left( \bigcup_{i=1}^d \left[ \left( \{ \delta_i^+ = 1 \} \cap \{ \theta_i \leq \theta_{\text{ref},i} \} \right) \cup \left( \{ \delta_i^- = 1 \} \cap \{ \theta_i \geq \theta_{\text{ref},i} \} \right) \right] \;\Bigg|\; \mathbf{y} \right).
\end{align*}
Controlling the BFWER at a level $\alpha$ ensures that the probability of making any erroneous directional claim in the entire set of comparisons is bounded by $\alpha$.
\end{definition}

The results established in Section~\ref{sec:comparaisonreference} assures that if the decision induced by $\delta$ is based on a credible rectangle at a level $1-\alpha$, then it controls the BFWER at a level $\alpha$. Thus the support estimator $\widehat{A}_n$ for $\alpha_n=\alpha$ controls the BFWER at a level $\alpha$.

\begin{remark}
By definition this does not provide a control on the usual frequentist FWER, the probability being controlled under two different set-up. However these two probabilities are linked, and a BFWER control still holds good FWER properties, as shown empirically in the next section.

\end{remark}

\section{Evaluation on a synthetic dataset}
\label{sec:resultssimu}

To evaluate the performance of the support estimator, we adopt a frequentist validation approach, comparing our proposed methods against standard frequentist estimators that control the Family-Wise Error Rate (FWER). The four estimators assessed are:
\begin{itemize}
    \item \textit{Bayes-Optimal}: Our proposed support estimator based on the optimal credible rectangle.
    \item \textit{Bayes-Bonferroni}: Our proposed support estimator based on the Bonferroni-type credible rectangle.
    \item \textit{Multiple-Test Bonferroni}: A multiple correlation test applied to each coefficient, corrected via the Bonferroni procedure.
    \item \textit{Multiple-Test Holm--Bonferroni}: A multiple correlation test applied to each coefficient, corrected via the Holm--Bonferroni procedure.
\end{itemize}

\paragraph*{Simulation Plan}
Using a ground-truth covariance matrix $\boldsymbol{\Sigma}_0$ and its corresponding support $A_0$, we simulate $n$ observations from $X \sim \mathcal{N}(0,\boldsymbol{\Sigma}_0)$ to estimate the support $\widehat{A}_0$. This process is repeated across $S=1000$ independent seeds to ensure robust performance evaluation. We assess the estimators using two primary metrics. First, we compute the mean accuracy, defined as the average proportion of correctly estimated coefficients across the simulations:
\begin{equation*}
    \text{Acc} = \frac{1}{S} \sum_{s=1}^{S} \left( \frac{| \widehat{A}_0^{(s)} \cap A_0 | + | (\widehat{A}_0^{(s)})^c \cap A_0^c |}{p(p-1)/2} \right) \times 100,
\end{equation*}
where $p(p-1)/2$ is the total number of coefficients in the support to estimate. To quantify the stability of this performance across seeds, we also report the standard deviation of the accuracy, that is linked to the variability during the simulation process of $X \sim \mathcal{N}(0,\boldsymbol{\Sigma}_0)$: 
\begin{equation*}
    \sigma_{\text{Acc}}^2 = \frac{1}{S-1} \sum_{s=1}^{S} \left( \text{Acc}^{(s)} - \text{Acc} \right)^2,
\end{equation*}
where $\text{Acc}^{(s)}$ is the accuracy for seed $s$ and ${\text{Acc}}$ is the mean accuracy. Second, to verify error control, we estimate the \textit{empirical FWER} as the proportion of simulations containing at least one false positive:
\begin{equation*}
    \widehat{\text{FWER}} = \frac{1}{S} \sum_{s=1}^{S} \mathds{1}\left(\exists (i,j) \in \widehat{A}_0^{(s)}, (i,j) \notin A_0 \neq \emptyset \right).
\end{equation*}

We evaluate these metrics across varying sample sizes $n \in \{50, 100, 500\}$ and three distinct sparse positive-definite ground-truth matrices $\boldsymbol{\Sigma}_0$ (Figure in appendix). These matrices were selected to represent varying edge densities ($2.4\%$, $24\%$, and $48\%$). The highly sparse configuration ($2.4\%$) reflects common assumptions in fMRI sparse modeling, while the denser configurations align with proportions observed in our real-world datasets. This design allows us to determine whether specific methods offer a superior balance between error control and edge detection power under different structural and asymptotic conditions.

\begin{table}[h]
\caption{Metrics obtained for each method based on $1000$ simulations for each configuration $(\boldsymbol{\Sigma}_0,n)$. Matrices $\boldsymbol{\Sigma}_0$ are characterized by their proportion of non-zeros coefficients (True Positive). Metrics used are mean accuracy Acc and empirical FWER $\widehat{\text{FWER}}$, both in percentages.}
\label{tab:performances_support}
\centering
\begin{tabular}{ccrlllllll}
\toprule
 & & \multicolumn{2}{c}{2.4\% True Positive} & \multicolumn{2}{c}{24\% True Positive} & \multicolumn{2}{c}{48\% True Positive} \\
\cmidrule(lr){1-2} \cmidrule(lr){3-4} \cmidrule(lr){5-6} \cmidrule(lr){7-8}
$n$ & Method & {\scriptsize ${\text{Acc}} \pm\sigma_{\text{Acc}}$} &  {\scriptsize $\widehat{\text{FWER}}$} & {\scriptsize ${\text{Acc}} \pm\sigma_{\text{Acc}}$} &  {\scriptsize $\widehat{\text{FWER}}$} & {\scriptsize ${\text{Acc}} \pm\sigma_{\text{Acc}}$} &  {\scriptsize $\widehat{\text{FWER}}$} \\
\midrule
\multirow{4}{*}{50} & Bonf. rectangle (our) & \textbf{98.2 (±0.1)} & 5.9 & 80.2 (±1.0) & 4.9 & 40.1 (±3.0) & 1.5 \\
  & Optimal rectangle (our) & \textbf{98.2 (±0.1)} & 7.8 & \textbf{80.5 (±1.0)} & 7.0 & \textbf{41.2 (±3.3)} & 2.4 \\
  & Mult. test Bonf. & \textbf{98.2 (±0.1)} & 4.8 & 79.8 (±0.9) & 3.5 & 38.9 (±2.8) & 1.1 \\
  & Mult. test Holm-Bonf. & \textbf{98.2 (±0.1)} & 4.9 & 79.8 (±0.9) & 3.5 & 39.1 (±2.9) & 1.2 \\

\addlinespace
\midrule
\addlinespace
\multirow{4}{*}{100} & Bonf. rectangle (our) & \textbf{98.7 (±0.1)} & 6.2 & 85.0 (±1.1) & 4.5 & 54.1 (±3.9) & 1.4 \\
  & Optimal rectangle (our) & \textbf{98.7 (±0.1)} & 6.6 & \textbf{85.2 (±1.1)} & 5.4 & \textbf{55.1 (±4.1)} & 1.8 \\
  & Mult. test Bonf. & \textbf{98.7 (±0.1)} & 5.2 & 84.8 (±1.1) & 3.9 & 53.3 (±3.8) & 1.1 \\
  & Mult. test Holm-Bonf. & \textbf{98.7 (±0.1)} & 5.3 & 84.8 (±1.1) & 4.2 & 54.1 (±4.0) & 1.5 \\

\addlinespace
\midrule
\addlinespace
\multirow{4}{*}{500} & Bonf. rectangle (our) & \textbf{99.4 (±0.1)} & 4.8 & 94.7 (±0.5) & 4.3 & 84.8 (±1.7) & 1.8 \\
  & Optimal rectangle (our) & \textbf{99.4 (±0.1)} & 5.3 & \textbf{94.8 (±0.5)} & 5.0 & 85.1 (±1.7) & 2.0 \\
  & Mult. test Bonf. & \textbf{99.4 (±0.1)} & 5.0 & 94.7 (±0.5) & 4.1 & 84.8 (±1.7) & 1.7 \\
  & Mult. test Holm-Bonf. & \textbf{99.4 (±0.1)} & 5.1 & \textbf{94.8 (±0.5)} & 5.0 & \textbf{85.7 (±1.7)} & 3.1 \\
\bottomrule
\end{tabular}

\end{table}

\paragraph*{Results}
The results in Table~\ref{tab:performances_support}  show that the Bayes-Optimal method has the best accuracy in almost every combination of  $\boldsymbol{\Sigma}$ and $n$, and it advantage increases with the proportion of edges for $n$ fixed. First, when the proportion of edges is high, Bayes-Optimal is much more efficient. This is linked to the dependencies between the coefficients, which are stronger when there are a lot of edges and are eliminated by the zero coefficients. The Bayes-Optimal method is the only one of the four estimators to take these dependencies into account. The same reasoning explains why, on the other hand, the Bayes-Optimal method becomes nearly identical to the other three methods when the proportion of edges is low. A second observation from this table is that Bayesian estimators are more accurate when $n$ is low. This is because the prior distribution stabilises the estimation in situations where the empirical estimator is very unstable. Overall, the gain in accuracy is very low or negligible when $n = 500$, but datasets of this length are unlikely in fMRI studies. 

Finally, the empirical FWER is slightly higher for Bayes-Optimal when $n$ is low, but similar to the results obtained using the Holm--Bonferroni correction. It should be stressed that the Bayesian method controls a Bayesian FWER, which is not what we are estimating here. 

In conclusion, we believe that this simulation study demonstrate the effectiveness of our estimator in identifying significant edges while maintaining control over the BFWER, which tends to be less conservative than frequentist methods. The increased flexibility is due to the fact that the optimal rectangle takes into account the multivariate dependencies. However, this is unnecessary if $n$ is high, as procedures such as the Holm--Bonferroni procedure can already detect most edges with much lower time and memory complexity.

\section{Application to fMRI datasets}
\label{sec:realdatasets}
Here, we apply our approach (model and methodology) to well-studied datasets. This serves both to validate our results against existing findings but also to demonstrate the enhanced individual-level interpretability offered by our method.

In this section we use the inverse-Wishart model presented in Section~\ref{sec:model} with the optimal rectangle construction. The methodology used to assess significant differences between correlations is the one presented in Section~\ref{sec:comparaisontwoparameters}. This comparison is ultimately a comparison of correlation matrices, but we use it to evaluate the difference between two fMRI scans.

\subsection{Human Connectome Project dataset}

Reliability of graph metrics has frequently been studied using test-retest datasets, such as the HCP dataset, which provides data for 100 healthy controls scanned at two different times. Studies like \cite{termenon_reliability_2016} used indicators such as IntraClass Correlation (ICC) to show that for many graph metrics, within-subject variance is lower than between-subject variance. We reproduced this analysis using our method to detect significant differences. For each scan $S_i^j$ (where $i$ is the subject and $j=1,2$ is the scan number), we compute the optimal rectangle for the correlation matrix at a $95\%$ level. If two rectangles are disjoint, the fMRI scans are considered significantly different. We then calculate the proportion of significant differences within subjects (between $S_i^1$ and $S_i^2$) and between subjects (for $i_1 \neq i_2$, between $S_{i_1}^{j_1}$ and $S_{i_2}^{j_2}$). This proportion represents the frequency where at least one disjoint interval exists between two scans, indicating a significant global difference; we do not evaluate the number of differing edges.

The proportion of times two scans of the same subject differ is $41\%$, whereas the between-subject difference rate is $83\%$. This aligns with known results on this dataset \citep{termenon_reliability_2016}. Our approach is particularly interesting because it assesses whether differences are significant, whereas previous comparisons were more arbitrary. However, these results also confirm that fMRI connectivity varies due to numerous factors, including within-individual differences. We do not claim to solve this issue; our work focuses on matrix inference, and the inherent challenges of interpreting correlations remain.

\subsection{Coma dataset}

This dataset provides access to fMRI scans of $10$ control volunteers and $10$ patients who experienced a brain injury resulting in a coma. Both groups were scanned in the same setting using a brain atlas divided into $p=107$ ROIs, with wavelet analysis providing $n=95$ iid observations. This dataset is of interest because each patient was scanned twice: once 30 days after the brain injury, and again 60 days later, allowing for longitudinal analysis. Further details can be found in \cite{oujamaa_functional_2023}. 

Previous studies have compared the control group with the patient group after 30 days and again after 60 days. They concluded that patients were closer to the healthy controls after 60 days than after 30 days, indicating a recovery period for the patients. However, this approach disregarded the individual evolution of each patient. In this section, our aim is to obtain similar results at the group level and to provide new tools for longitudinal follow-up of each patient.

\paragraph*{Comparison to the control group} To mimic the results of previous studies on this dataset, the first step is to use our method to compare patient scans to the control group. It is standard to use for the control group a matrix $\boldsymbol{C}_{\text{control}}$ which is the mean of the correlation matrices of the volunteers. In our case, we want to introduce uncertainty into this matrix as it was computed on a low number of healthy subjects, so we consider a posterior distribution, $\mathcal{IW}(\boldsymbol{C}_{\text{control}},\nu+95)$, for the control group. This allows us to reproduce, for this control matrix, the uncertainty that we have for each patient with $n=95$ observations. We then compare the credible region of this posterior distribution to the credible region at the $95 \%$ level for each patient for both scans. The results are summarised in Table~\ref{tab:tablecoma}. %

\begin{table}[h]

\centering
\caption{Cousciousness of each patient (MCS for \emph{minimally conscious state}, and C for \emph{coma}) and comparison of credible rectangles based on the number of disjoint edges for each patient for both exams with the distribution $\mathcal{IW}(\boldsymbol{C}_{\text{control}},\nu+95)$, based on a group of control subjects. Bold font is used for patients whose conscious state has changed.}
\label{tab:tablecoma}
\begin{tabular}{ccccc}
\toprule
{Patient} & \multicolumn{2}{c}{Consciousness state} & \multicolumn{2}{c}{Comparison to $\boldsymbol{C}_{\text{control}}$} \\
\cline{2-5}
 & Exam 1 & Exam 2 & Exam 1 & Exam 2 \\
\midrule
\textbf{3}  & \textbf{MCS} & \textbf{C} & \textbf{14}  & \textbf{7}   \\
4  & C & C & 7   & 0   \\
7  & C & C & 3   & 6   \\
8  & C & C & 8   & 0   \\
\textbf{9}  & \textbf{MCS} & \textbf{C} & \textbf{40}  & \textbf{3}   \\
15 & C & C & 2   & 5   \\
19 & C & C & 12  & 21  \\
\textbf{22} & \textbf{MCS} & \textbf{C} & \textbf{164} & \textbf{5}   \\
25 & C & C & 3   & 2   \\
29 & C & C & 9  & 5   \\
\bottomrule
\end{tabular}

\end{table}

We are particularly interested in how the patients' scans evolve between the first and second scans compared to the control group. Analysis of Table~\ref{tab:tablecoma} reveals a general trend of decreased differences in the second examination compared to the first. Of the $10$ patients, seven exhibited a reduction in the metric, with two of them reaching a value of zero. This tendency was even more pronounced among the three patients diagnosed with a minimally conscious state (MCS). Notably, Patients 9 and 22, who displayed substantial deviations from the control state during the initial exam, showed a significant reduction in these differences by the second exam. However, it is important to note that some patients deviated from this trend; Patient 19 in particular warrants further attention.

\paragraph*{Longitudinal analysis of a patient}
The richness of our method lies in its ability to directly compare two scans of the same patient, as demonstrated with the previous dataset. When we applied this to patients in the coma dataset, we found significant differences between the two scans for each patient, resulting in a $100\%$  within-subject difference. Combined with the previous results in Table~\ref{tab:tablecoma}, this result confirms our ability to detect changes induced by the recovery of patients. 

\begin{figure}[h]
    \centering
    \begin{subfigure}[b]{0.49\textwidth}
        \centering
        \includegraphics[width=\textwidth]{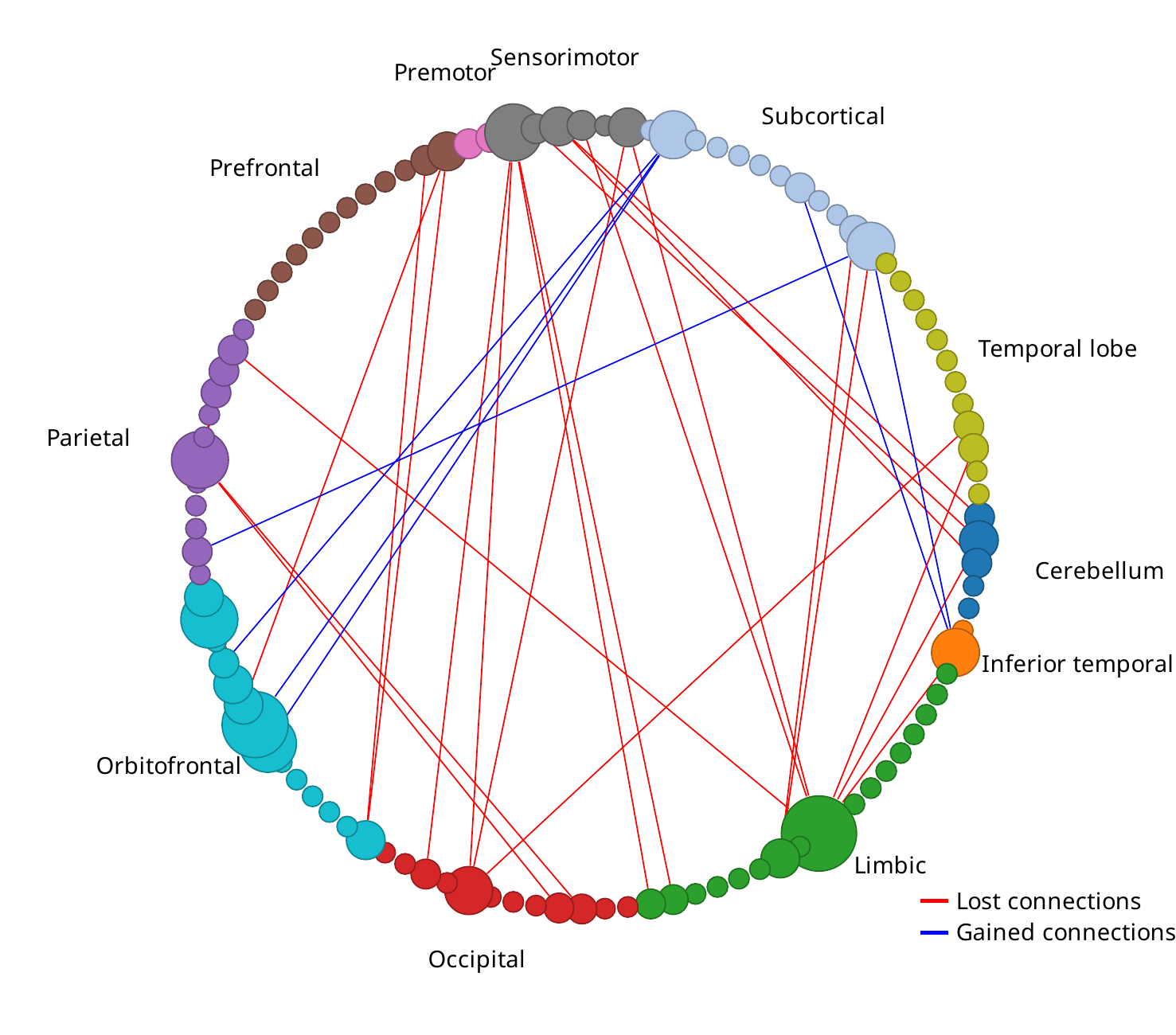}
        \subcaption{Patient 9.}
    \end{subfigure}
    \hfill
    \begin{subfigure}[b]{0.49\textwidth}
        \centering
        \includegraphics[width=\textwidth]{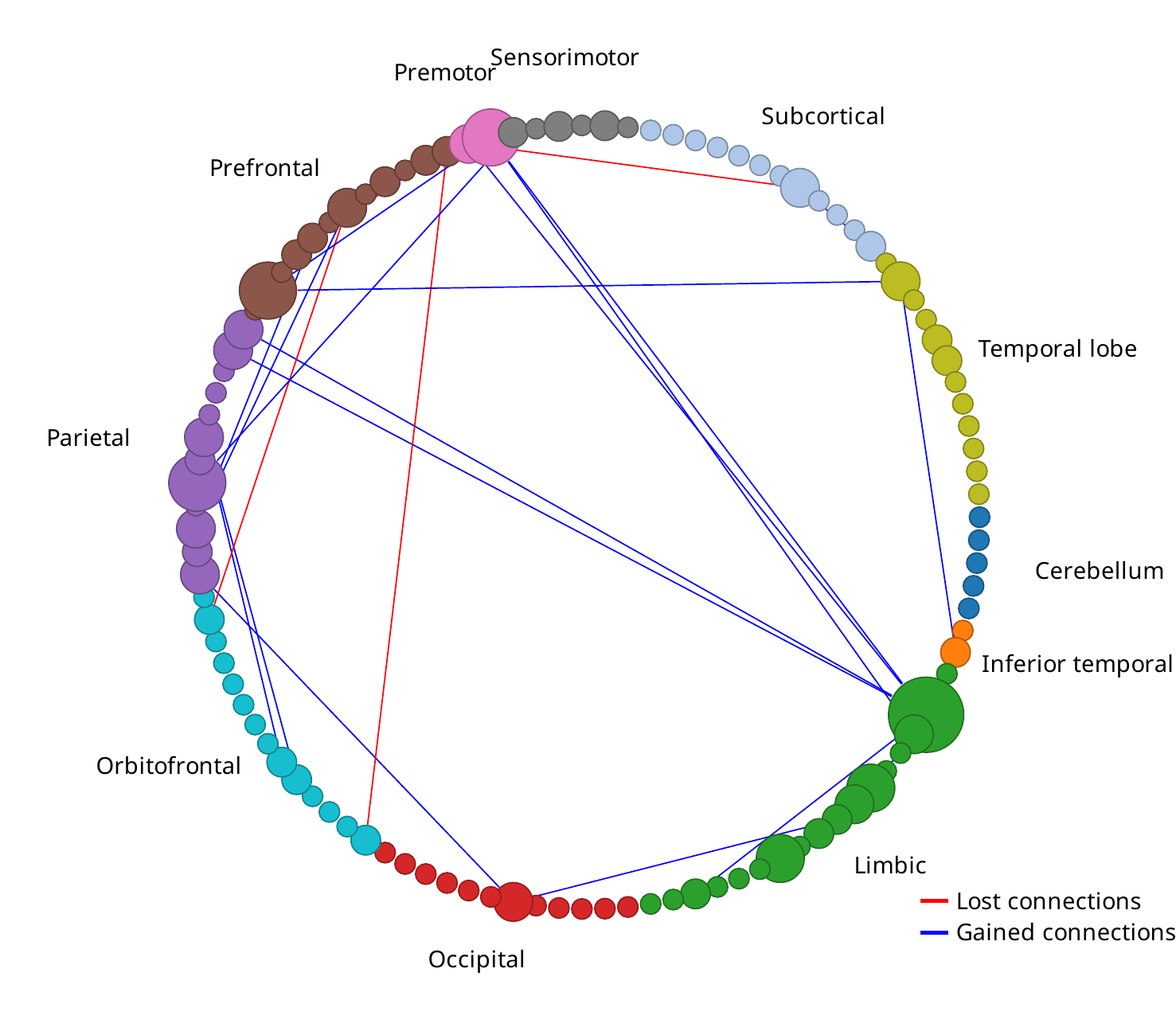}
        \subcaption{Patient 22.}
    \end{subfigure}
    \caption{Network representations of the gained and lost connections, respectively in blue and red links, between the first and second exam for two patients. The ROIs are reorganized in subnetworks. A connection is considered gained (resp. lost) if it is significantly greater (resp. lower) in the second exam compared to the first.}
    \label{fig:chords}
\end{figure}

We can take a closer look at the representation of the credible rectangles. Specifically, we will focus on Patients 9 and 22, who showed recovery between the two scans. Using the comparison between the first and second exams, we can identify where the differences occur. The methodology from Section~\ref{sec:comparaisontwoparameters} allows us to find the connections that were significantly affected between these two scans, as well as whether they were affected negatively or positively, by comparing the marginal intervals. To illustrate this, we have chosen to present them in the form of a chord diagram, reorganising the regions into subnetworks. An example can be seen in Figure~\ref{fig:chords}, where these diagrams are represented for Patients 9 and 22. The edges that appear represent those with significant differences between the first and second exams, meaning these marginal intervals were disjoint. The colour of the edges represents the direction of evolution: if the interval of the second exam is greater, the connection has appeared (blue); if it is lower, the connection has been lost (red). The interesting aspect of this figure is how the two patients' evolutions are completely different from a connectivity graph perspective. Patient 9 undergoes an overall loss of connectivity, whereas Patient 22 gains more connections. The affected subnetworks are also different. This supports our view that individual follow-up is essential.

\paragraph*{Raising awareness of uncertainty}

\begin{figure}[h]
    \centering
    \begin{subfigure}[b]{0.3\textwidth}
        \centering
        \includegraphics[width=\textwidth]{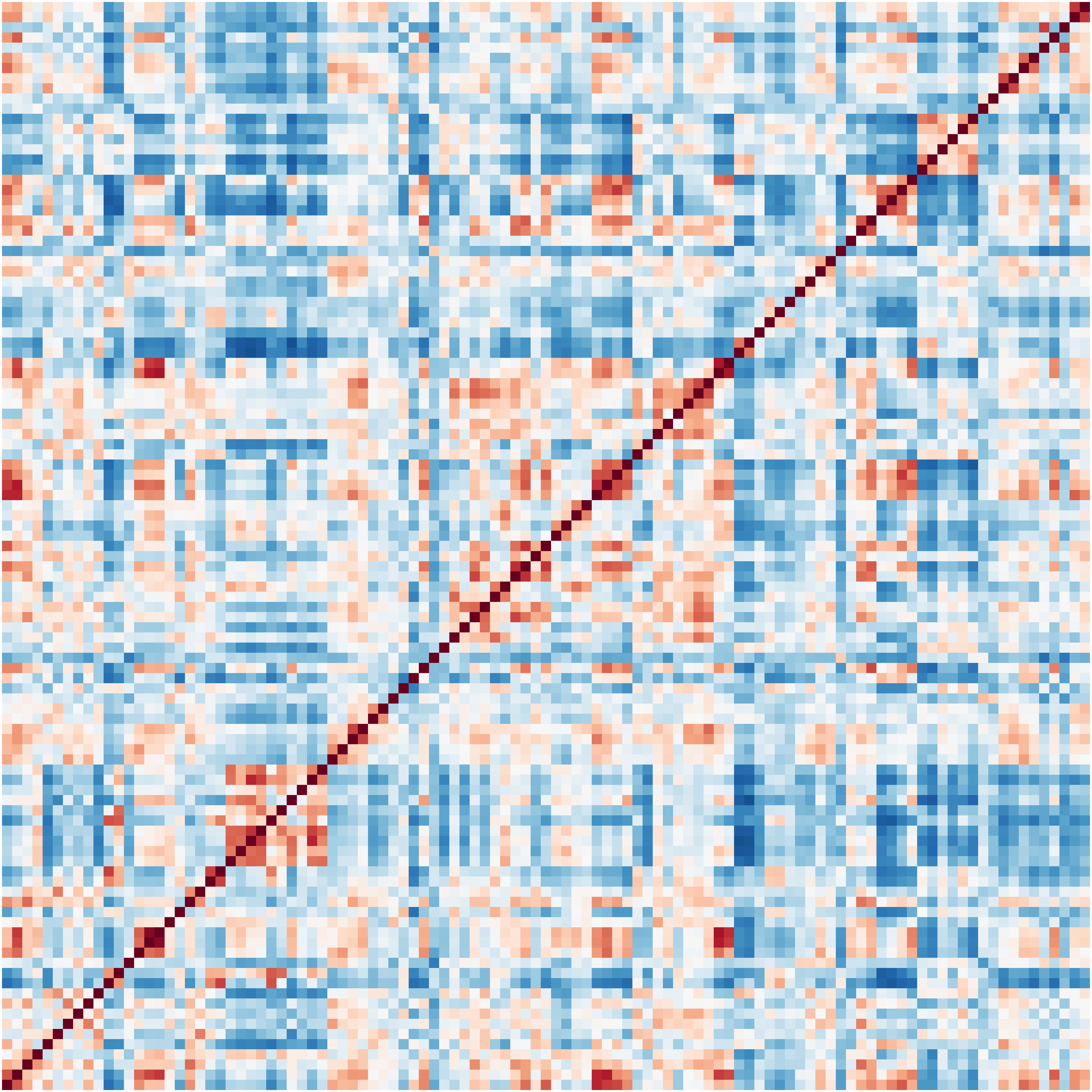}
        \subcaption{Lower Quantile 1}
    \end{subfigure}
    \hfill
    \begin{subfigure}[b]{0.3\textwidth}
        \centering
        \includegraphics[width=\textwidth]{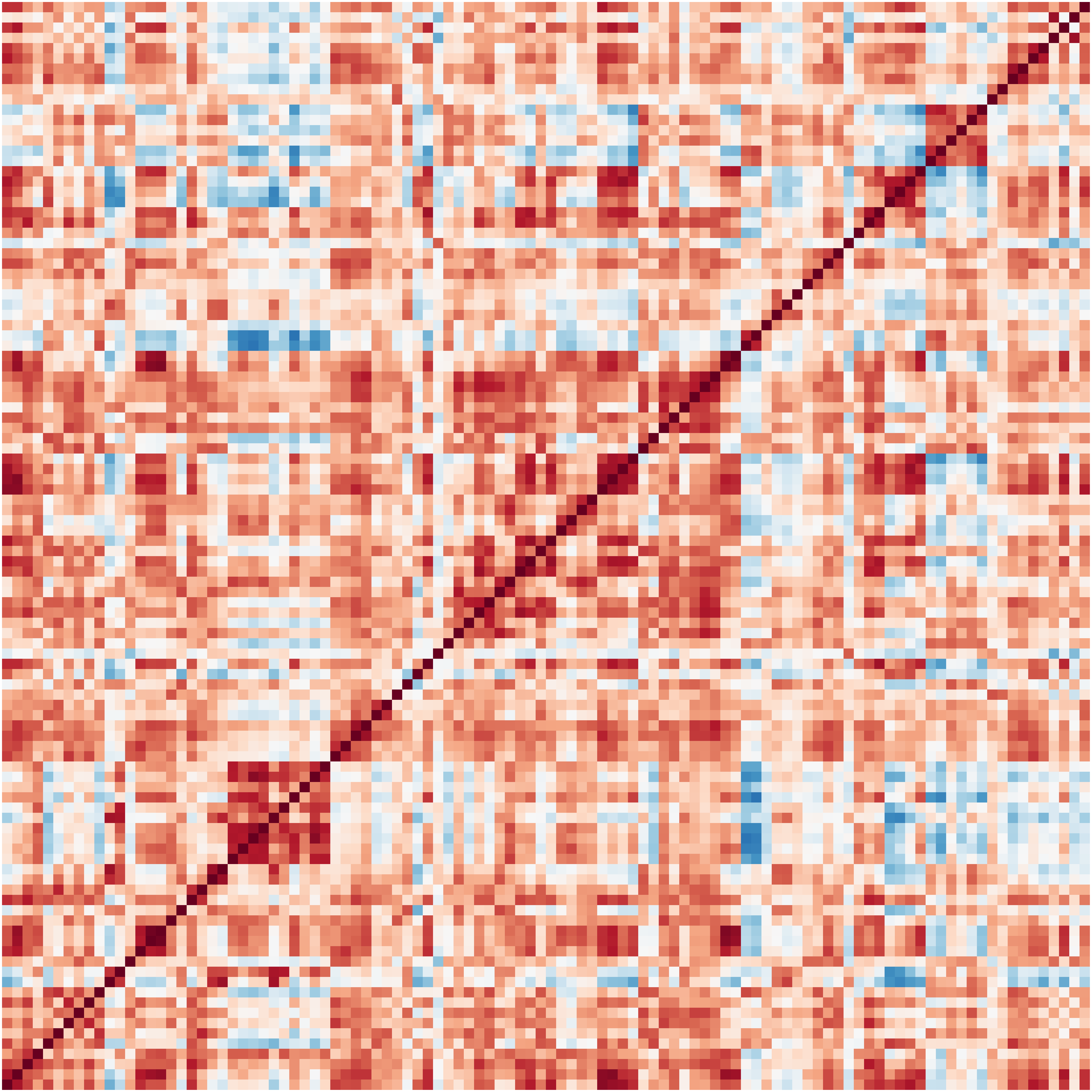}
        \subcaption{$\hat{C_n}$ Exam 1}
    \end{subfigure}
    \hfill
    \begin{subfigure}[b]{0.3\textwidth}
        \centering
        \includegraphics[width=\textwidth]{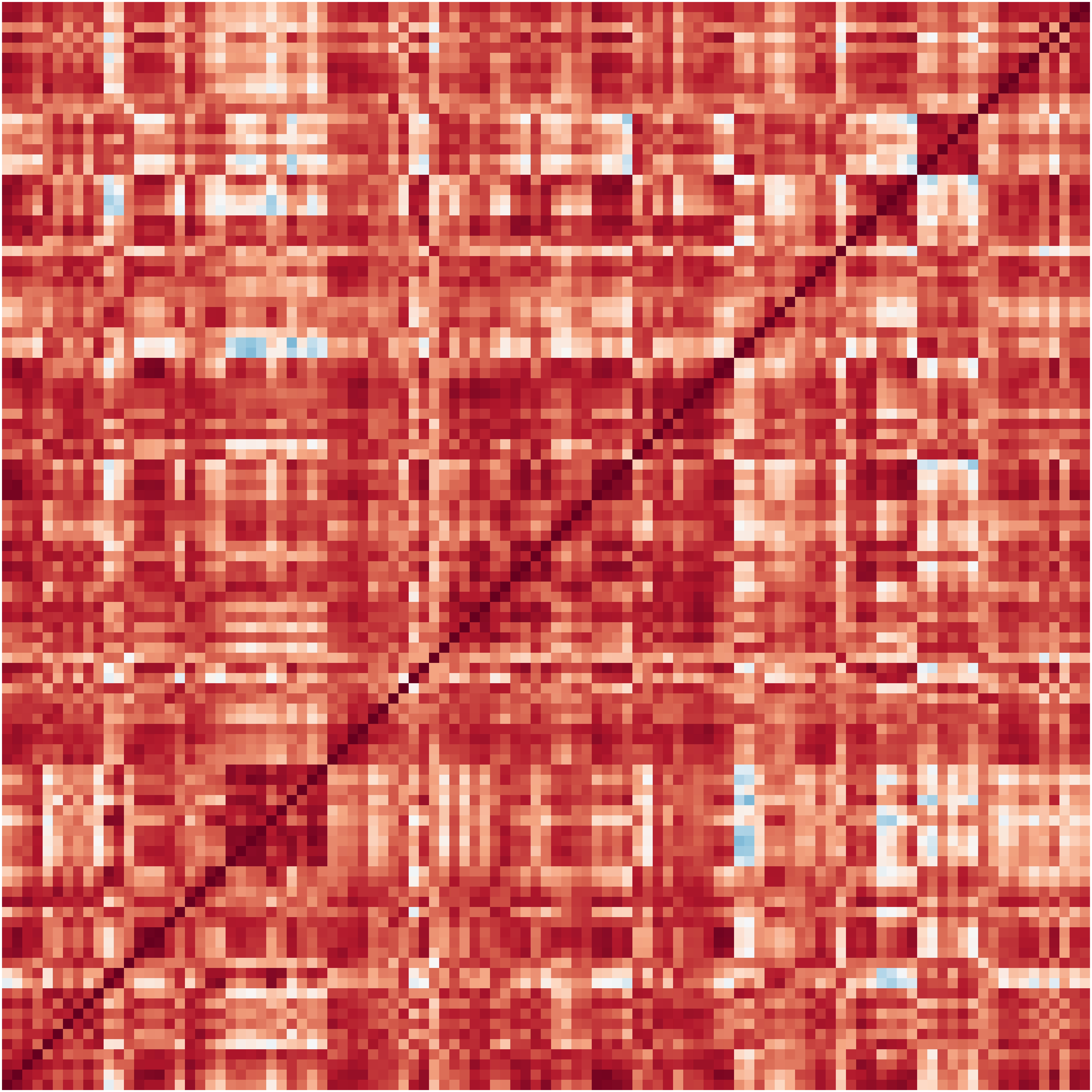}
        \subcaption{Upper Quantile 1}
    \end{subfigure}
    
    \vspace{0.5cm}
    
    \begin{subfigure}[b]{0.3\textwidth}
        \centering
        \includegraphics[width=\textwidth]{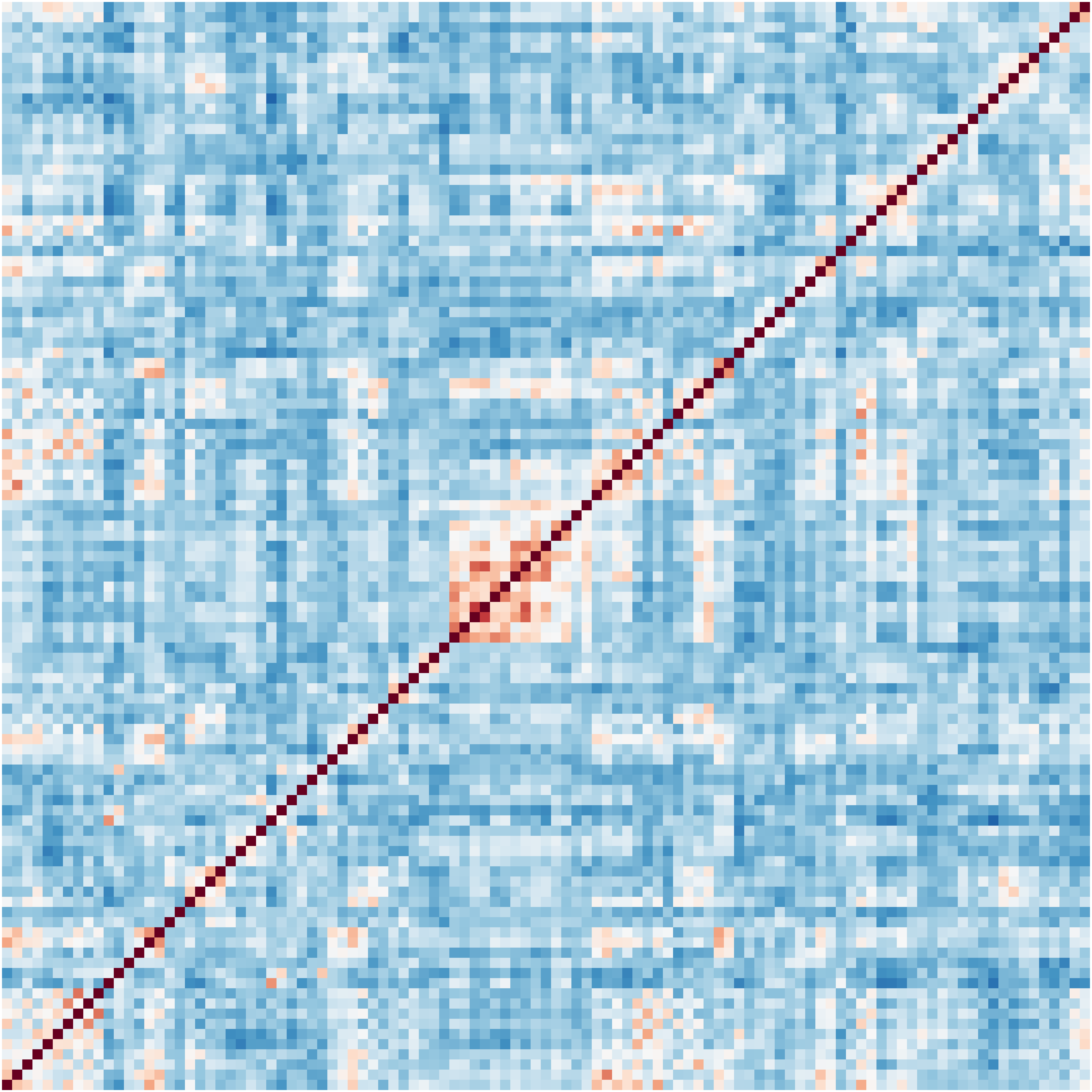}
        \subcaption{Lower Quantile 2}
    \end{subfigure}
    \hfill
    \begin{subfigure}[b]{0.3\textwidth}
        \centering
        \includegraphics[width=\textwidth]{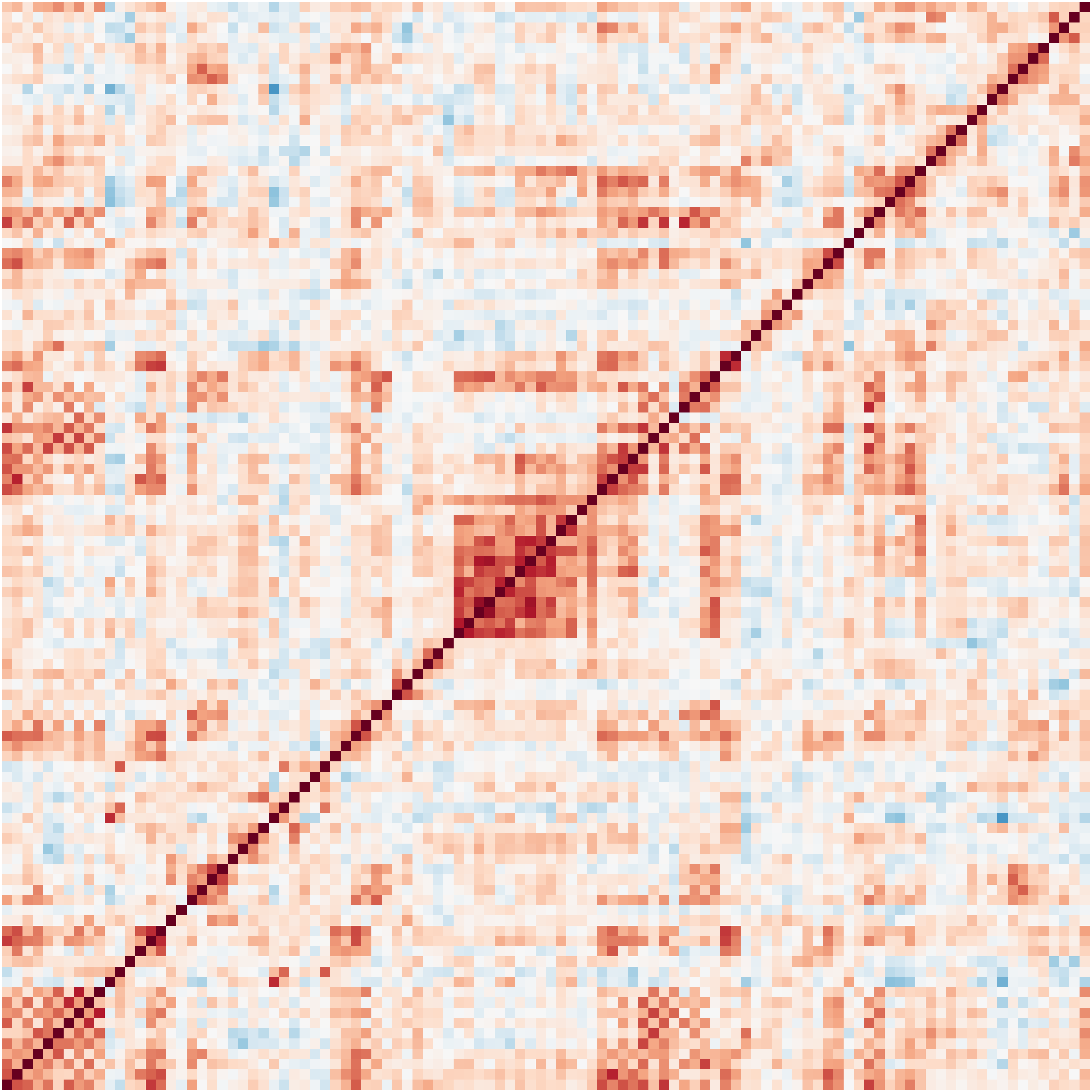}
        \subcaption{$\hat{C_n}$ Exam 2}
    \end{subfigure}
    \hfill
    \begin{subfigure}[b]{0.3\textwidth}
        \centering
        \includegraphics[width=\textwidth]{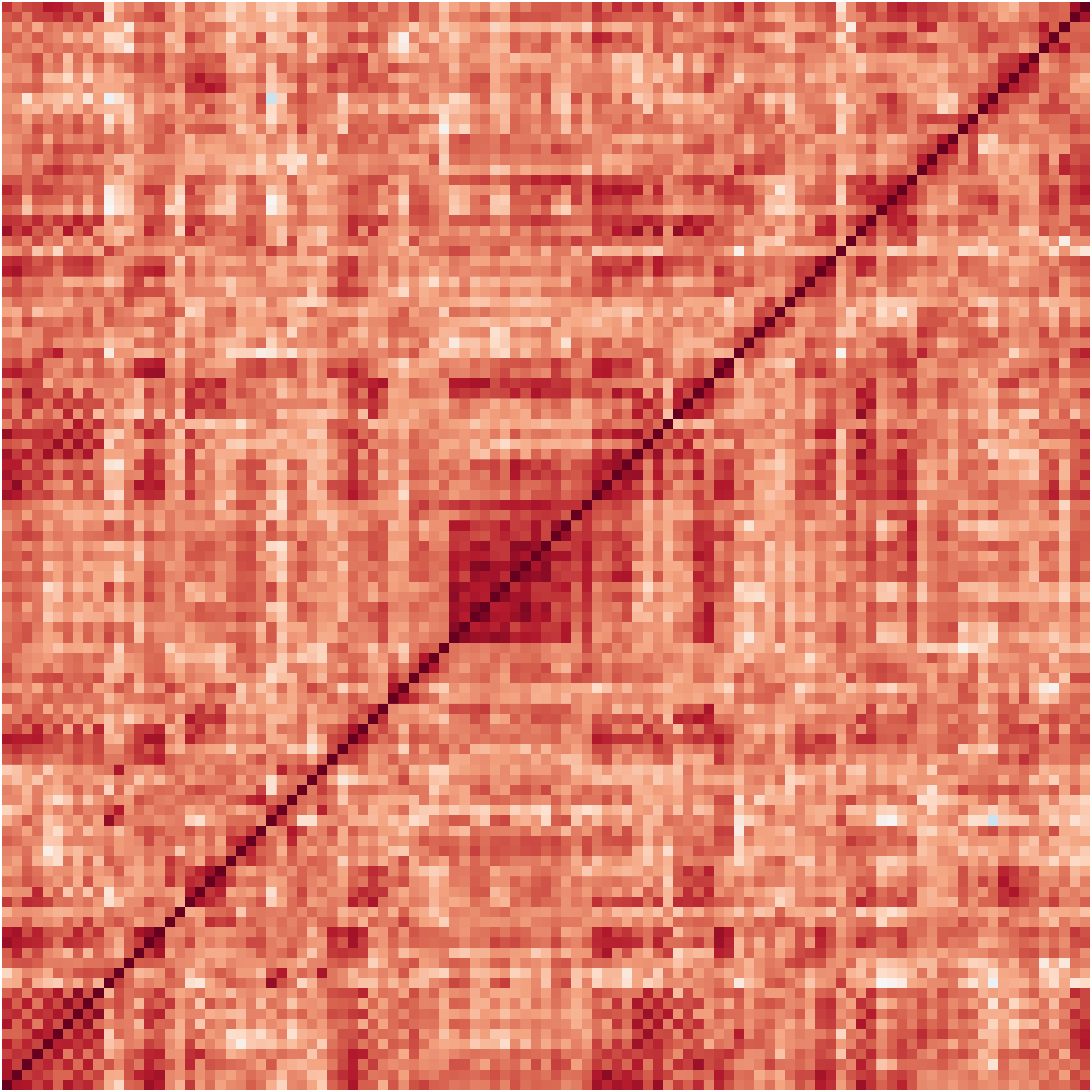}
        \subcaption{Upper Quantile 2}
    \end{subfigure}
    
    \caption{Empirical correlation matrices and quantiles representation (lower and upper) of each credible rectangle for the subject 9 at two different time (Exam 1: $30$ days after coma in first row, Exam 2: $60$ days after coma).}
    \label{fig:patient9_matrices}
\end{figure}

Thanks to their rectangular shape, it is possible to represent the credible regions directly. Each credible rectangle provides two matrices for the upper and lower quantiles. While we have seen how this can be used to detect differences, the matrices can also be used to visualise uncertainty. Figure~\ref{fig:patient9_matrices} shows both the empirical correlation matrices and the upper and lower quantile matrices for Patient 9 in both exams. This help us identify visually some zones that are highly correlated with high certainty; for instance in the second row we detect a red square in the middle that corresponds to the visual cortex, with significantly non-zero coefficients (red in the lower matrix). This also illustrates the unstable nature of correlation matrix estimation based on a small number of realisations for a high-dimensional matrix. We believe this carries a strong message in itself as it raises awareness of instability. It encourages caution when working with such estimations.

\subsection{Trade-off between data size and uncertainty}

As demonstrated throughout this paper, the credible region size depends heavily on $n$, the number of data realizations. However, in fMRI acquisition, the number of realizations is limited by exam costs and patient discomfort during prolonged scans. Consequently, increasing the number of data points is often unfeasible. Alternatively, the credible region size is also influenced by the random variable's dimension and the dependencies between coefficients. It is instructive to examine how the coefficients of $\boldsymbol{C}$, alongside $n$ and $p$, impact this size. Since the volume of rectangle of different dimensions are not comparable, we chose to compare them using the mean length of the marginal intervals instead. For a rectangle $\mathcal{R}=\prod_{i<j}[q_{ij}^l,q_{ij}^u]$, the mean length is defined as:
\begin{align}
\label{eq:meanlength}
    m_l(\mathcal{R})=\frac{\sum_{i<j} (q_{ij}^u-q_{ij}^l)}{p(p-1)/2}.
\end{align}

We evaluate $m_l$ for various inverse-Wishart distributions parametrized by $(n,p,\rho)$, where $n$ denotes the number of observations, $p$ the matrix dimension, and $\rho$ characterizes dependencies via $\boldsymbol{\Sigma}(\rho)$ as in Equation~\ref{eq:inversewishartexemple} with the distribution $\mathcal{IW}_p(\boldsymbol{\Sigma}(\rho),p+2+n)$.

\begin{figure}[h]
    \centering
    \includegraphics[width=0.99\linewidth]{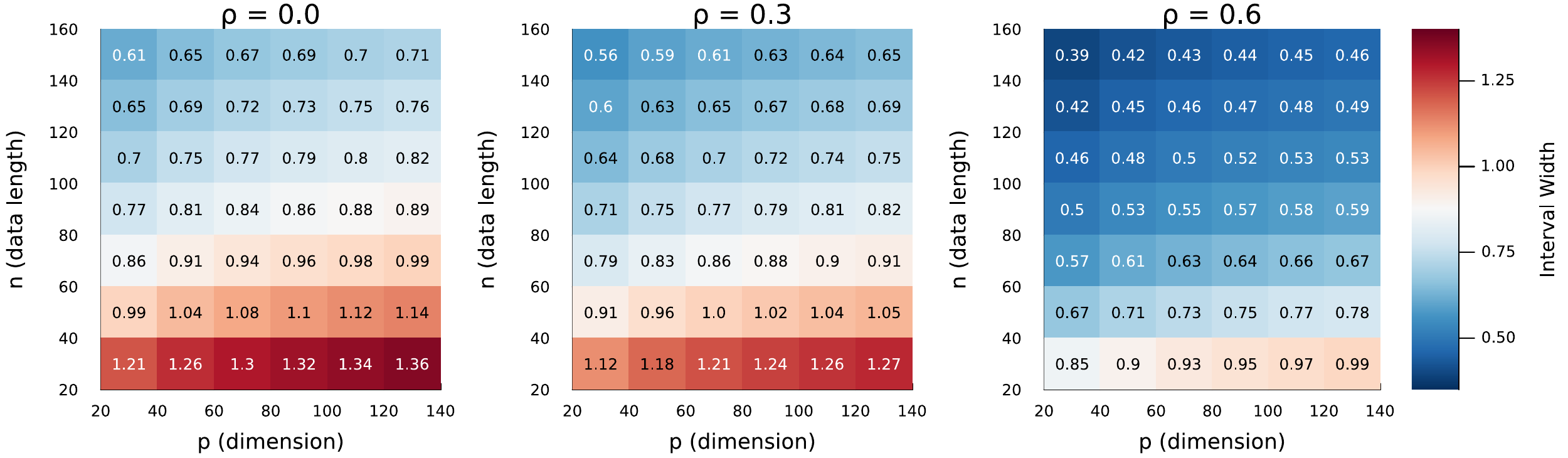}
    \caption{Heatmaps of the mean length $m_l$ defined in Equation~\ref{eq:meanlength} as a function of the parameter $\rho$, number of variables $p$, and number of time points $n$. The mean length is computed for $95\%$ credible rectangles derived for the distribution of correlation coefficients based on the distribution $\mathcal{IW}_p(\boldsymbol{\Sigma}(\rho),p+2+n)$ (defined in Equation~\ref{eq:inversewishartexemple}). Each cell represents the average interval length over 100 simulations.}
    \label{fig:heatmap_interval_length}
\end{figure}

Figure~\ref{fig:heatmap_interval_length} displays these results as a grid of $n$ and $p$ for fixed $\rho$. The results confirm the expected behavior: length decreases as $n$ or $\rho$ increases, and increases with $p$. However, the length tends to plateau as $n$ decreases, following an initial sharp gain. This suggests that a moderate increase in scan length can efficiently decreases uncertainty.

We present this grid primarily as a guide for compromise. Indeed, to detect significant differences in correlation (e.g., a difference of $0.4$), one could focus on a specific subnetwork of size $p'<p$ using this grid. Reducing the dimension increases statistical power. We argue that the abundance of information in fMRI scans often misleads researchers into analyzing the whole brain; however, focusing on specific brain subnetworks may yield more robust conclusions.

\section{Discussion}

The primary advantage of a Bayesian framework for modeling correlation matrices lies in accessing the joint distribution of coefficients, a task rarely feasible in a non-asymptotic frequentist framework. We illustrated this benefit through FWER simulations, where frequentist procedures were proved to be over-conservative due to their inability to account for dependencies between coefficients. Although computing the joint distribution incurs higher computational costs, this approach serves as a compelling proof of concept, motivating a paradigm shift toward uncertainty-aware inference.

Our central objective is the direct quantification and visualization of this uncertainty, particularly when applied to real fMRI data. This application provides tangible tools for interpreting longitudinal patient evolution, offering critical interpretability at the single-subject level. While larger studies would be required to validate these findings in a definitive clinical trial, the current framework is already encouraging in itself. Extending this core goal, the grid presented in Figure~\ref{fig:heatmap_interval_length} quantifies uncertainty as a function of $(p,n)$ to guide the design of future clinical studies. By exposing the mathematical limits of uncertainty reduction, this grid demonstrates how to optimize study protocols, specifically by reducing the number of regions of interest to focus on biologically relevant subnetworks, thereby maximizing statistical power within feasible acquisition constraints. This encourages researchers to adopt modest, mathematically grounded objectives rather than attempting to overcome intrinsic high-dimensional limitations.

We recognize that fMRI connectivity analysis encompasses more than just correlation matrix inference. Nevertheless, this study highlights the inherent instability of these matrices, cautioning against their uncritical use in deriving graph metrics. A key direction for future work involves propagating this quantified uncertainty into tools specifically designed for graph analysis.

While our model is relatively simple, a choice that may invite criticism regarding potential bias or hyperparameter sensitivity, we argue that its behavior (specifically, slight shrinkage toward zero) aligns with several existing methods. The originality of this work does not lie in proposing a superior estimator, but in establishing a general, adaptive methodology to quantify and exploit estimation uncertainty. This approach informs decision-making, even if it cannot fully resolve the instability inherent in high-dimensional settings. This methodology remains adaptable to any model allowing posterior sampling.

\section*{Acknowledgments}
This work was partly supported by the Agence Nationale de la Recherche under the France 2030 programme, references ANR-23-IACL-0006 and ANR-21-JSTM-0001.

The first dataset was provided by the Human Connectome Project, WU-Minn Consortium (Principal Investigators: David Van Essen and Kamil Ugurbil; 1U54MH091657) funded by the 16 NIH Institutes and Centers that support the NIH Blueprint for Neuroscience Research; and by the McDonnell Center for Systems Neuroscience at Washington University. Full project accessible at
\url{https://www.humanconnectome.org/study/hcp-young-adult}.

The second dataset is based on a coma patients study that was conducted at Grenoble Alpes University Hospital between February 2015 and March 2018. The protocol was approved by the Ethics committee of the University Hospital of Grenoble Alpes, France (Comité Consultatif pour la Protection des Personnes Sud Est V, ID-RCB 2014-A01873-44/1). The study was registered on ClinicalTrials.gov (Identifier: NCT02647996).

\appendix
\section{Supplementary Materials}

To ensure the reproducibility and completeness of this study, we provide complementary resources:

\begin{itemize}
    \item \textit{Appendix:} detailed mathematical proofs for all propositions and theorems, along with additional figures that complement the main text.
    
    \item \textit{Open-Source Repository:} A public GitLab repository hosting the full source code and data required to reproduce our work. This includes:
    \begin{itemize}
        \item The raw and preprocessed datasets used in the analysis.
        \item All scripts implementing the proposed algorithms (Sliced-Quantile Grid and BGHM online).
        \item Pre-computed results to facilitate immediate inspection.
        \item Specific scripts to regenerate every figure and table presented in this paper.
    \end{itemize}
    The repository is available at: \url{https://gricad-gitlab.univ-grenoble-alpes.fr/chevaual/credible_rectangles_for_comparison_fmri}
\end{itemize}

\subsection{Additional figures}
We include additional figure and table of simulation and methods mentioned in the main text. 

\begin{figure}[ht]
\centering
 \begin{subfigure}{0.3\textwidth}
 \centering
     \includegraphics[width=4cm]{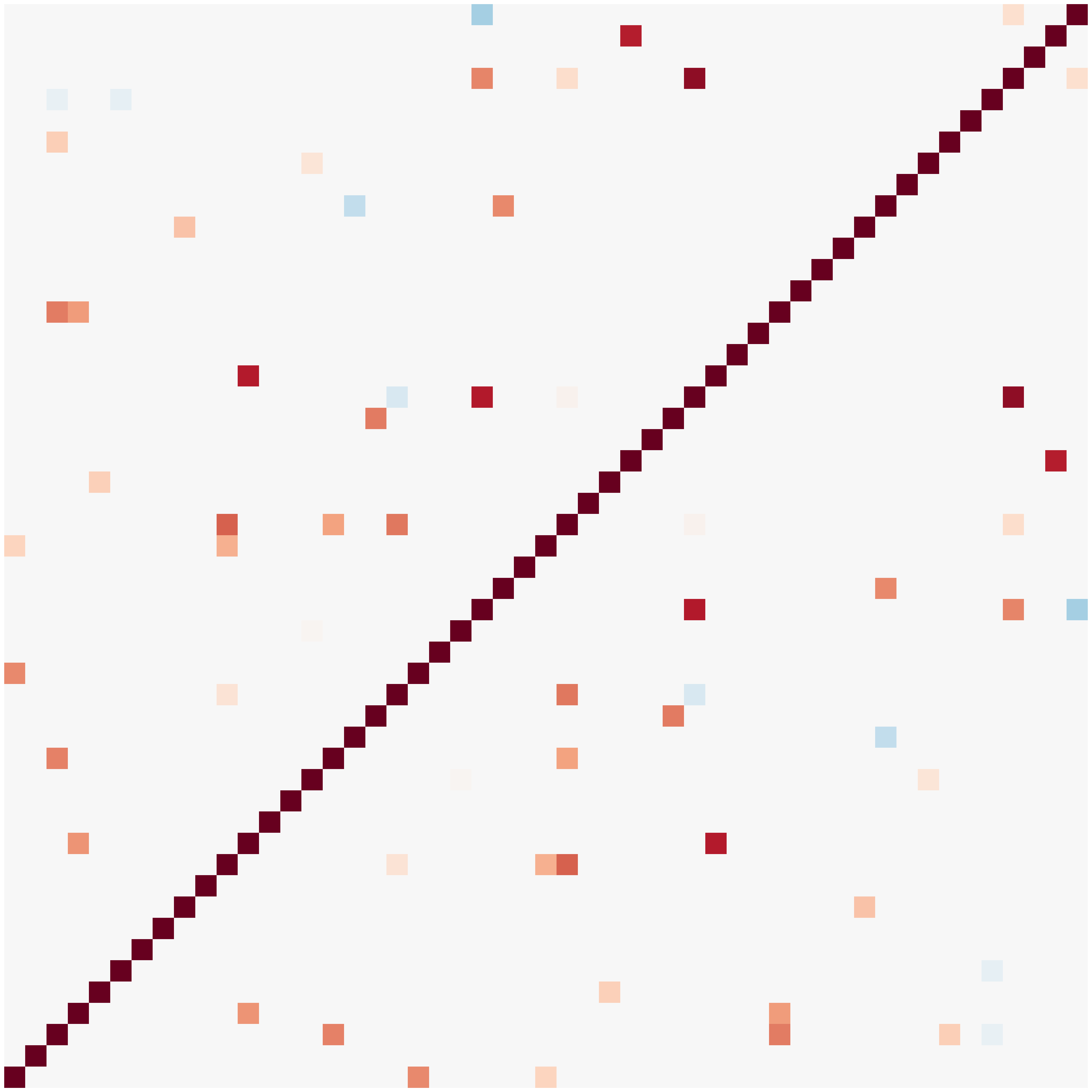}
     \caption{$2.4\%$ of edges}
     \label{fig:a}
 \end{subfigure}
 \hfill
 \begin{subfigure}{0.3\textwidth}
 \centering
     \includegraphics[width=4cm]{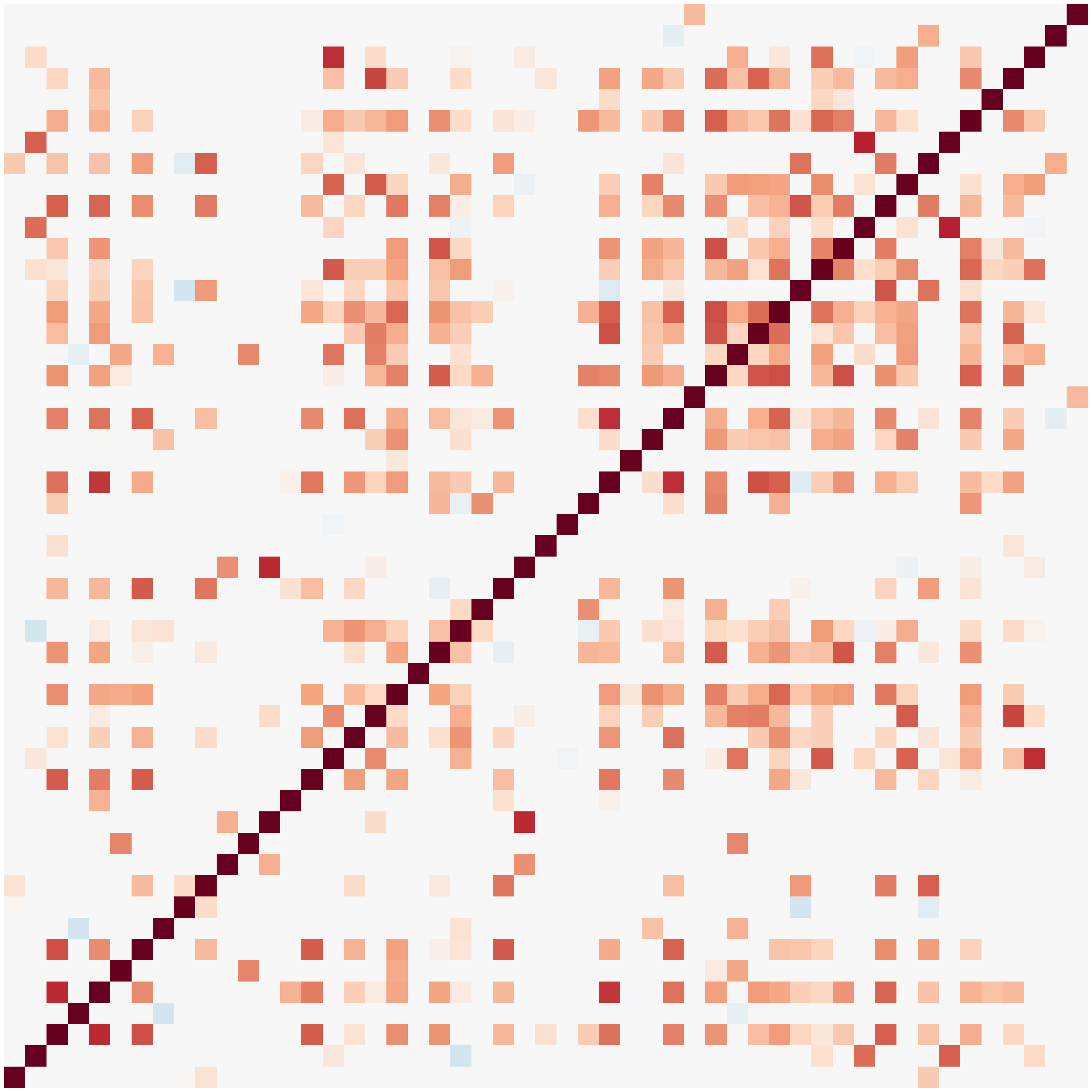}
     \caption{$24\%$ of edges}
     \label{fig:b}
     \hfill
 \end{subfigure}
 \hfill
 \begin{subfigure}{0.3\textwidth}
 \centering
     \includegraphics[width=4cm]{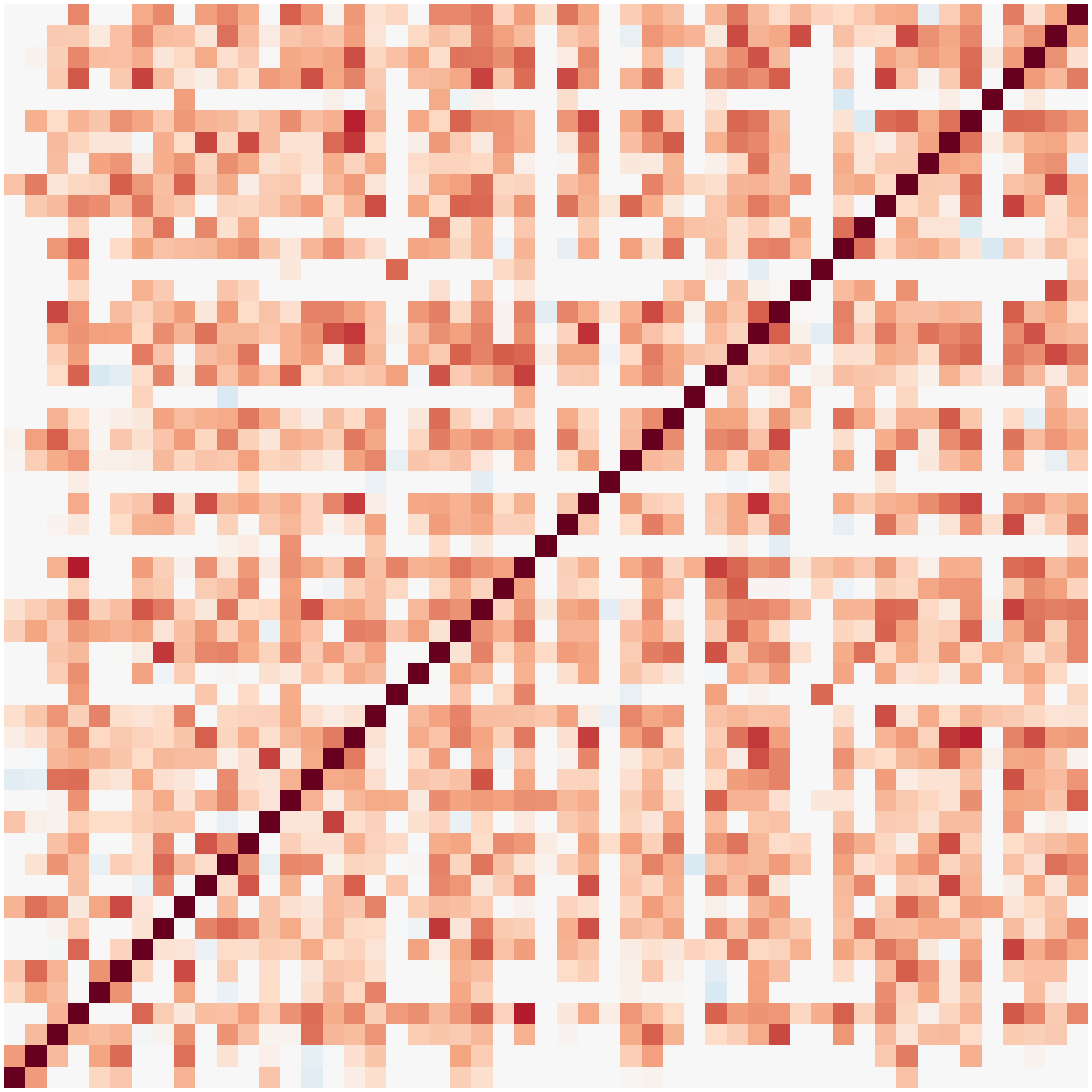}
     \caption{$48\%$ of edges}
     \label{fig:c}
 \end{subfigure}
 \caption{Ground-truth correlation matrices used for simulation study used in Section 5 for Table 2}
 \label{fig:heatmaps}

\end{figure}

\begin{table}[h]
\centering
\caption{Comparison consciousness state (CS) and distance $r$ defined in Proposition 3.2.1 between the credible region and the matrix of reference $\boldsymbol{C}_{control}$ for both exams}
\label{tab:results}
    \centering
\begin{tabular}{rrrrr}
\toprule
ID & CS 1 & CS 2 & Dist 1 & Dist 2 \\
\midrule
3  & MCS & C  & 3.05 & 3.04 \\
4  & C & C  & 3.53 & 1.82 \\
7  & C & C  & 2.36 & 2.81 \\
8  & C & C  & 3.36 & 0.72 \\
9  & MCS & C  & 5.43 & 3.13 \\
15 & C & C  & 3.28 & 3.14 \\
19 & C & C & 3.99 & 5.06 \\
22 & MCS & C & 9.38 & 3.35 \\
25 & C & C  & 2.30 & 1.30 \\
29 & C & C  & 3.13 & 2.23 \\
\bottomrule
\end{tabular}
\end{table}

\subsection{Proofs}

\begin{proof}[Proof of Proposition 3.2]
Since $R_{\theta}$ is a closed rectangle, it is a closed set in $\mathbb{R}^d$ under any norm. For any $\theta_{\text{ref}} \notin R_{\theta}$, the distance to the closed set is attained; thus, there exists $\theta^* \in R_{\theta}$ such that $\| \theta^* - \theta_{\text{ref}} \| = d$, with $d > 0$.

Consider the open ball $B_{\| \cdot \|}(\theta_{\text{ref}}, d) = \{ z \in \mathbb{R}^d : \| z - \theta_{\text{ref}} \| < d \}$. By definition of $d$ as the minimum distance, $B_{\| \cdot \|}(\theta_{\text{ref}}, d) \cap R_{\theta} = \emptyset$. This implies:
\[
\mathbb{P}(\| \theta - \theta_{\text{ref}} \| < d \mid \mathbf{y}) \leq \mathbb{P}(\theta \notin R_{\theta} \mid \mathbf{y}).
\]
Since $\mathbb{P}(\theta \in R_{\theta} \mid \mathbf{y}) \geq 1-\alpha$, we have $\mathbb{P}(\theta \notin R_{\theta} \mid \mathbf{y}) \leq \alpha$, which concludes the proof. 
\end{proof}

\begin{proof}[Proof of Proposition 3.3]
The posterior expected loss is the probability of making at least one directional error:
\[
\mathbb{E}_{\theta \mid \mathbf{y}}[L(\theta, \delta)] = \mathbb{P}\left( \bigcup_{i=1}^d \left( \{ \delta_i^+ = 1, \theta_i \leq \theta_{\text{ref},i} \} \cup \{ \delta_i^- = 1, \theta_i \geq \theta_{\text{ref},i} \} \right) \;\Bigg|\; \mathbf{y} \right).
\]
Substituting the definitions of $\delta_i^+$ and $\delta_i^-$, this probability becomes and we have:
\begin{align*}
    \mathbb{E}_{\theta \mid \mathbf{y}}[L(\theta, \delta)] &=\mathbb{P}\left( \bigcup_{i=1}^d \left( \{ \theta_{\text{ref},i} < q_i(t/2) \text{ and } \theta_i \leq \theta_{\text{ref},i} \} \cup \{ \theta_{\text{ref},i} > q_i(1-t/2) \text{ and } \theta_i \geq \theta_{\text{ref},i} \} \right) \;\Bigg|\; \mathbf{y} \right)\\
    &= \mathbb{P}\left( \bigcup_{i=1}^d \left( \{ \theta_i < q_i(t/2) \} \cup \{ \theta_i > q_i(1-t/2) \} \right) \;\Bigg|\; \mathbf{y} \right)\\
    &=\mathbb{P}( \theta \notin R(t) \mid \mathbf{y} ).
\end{align*}

Since $R(t)$ is a credible region at level $1-\alpha$, we have $\mathbb{E}_{\theta \mid \mathbf{y}}[L(\theta, \delta)] \leq \alpha$. 
\end{proof}

\begin{proof}[Proof of Proposition 3.4]
Let $R_A$ and $R_B$ be $(1-\alpha)$ credible regions for $\theta_A$ and $\theta_B$ respectively, with $R_A \cap R_B = \emptyset$.
By definition of credible regions:
\begin{align*}
\mathbb{P}(\theta_A \in R_A | \textbf{y}) &\geq 1-\alpha \\
\mathbb{P}(\theta_B \in R_B | \textbf{y}) &\geq 1-\alpha
\end{align*}

If $\theta_A = \theta_B = \theta^*$, then for equality to hold we must have $\theta^* \in R_A \cap R_B$. But since $R_A \cap R_B = \emptyset$, we have:
\[
\{\theta_A = \theta_B\} \subseteq \{\theta_A \notin R_A \} \cup \{ \theta_B \notin R_B\}
\]

Therefore:
\[
\mathbb{P}(\theta_A = \theta_B | \textbf{y}) \leq \mathbb{P}(\{\theta_A \notin R_A \} \cup \{ \theta_B \notin R_B\} | \textbf{y})
\]

By the union bound:
\[
\mathbb{P}(\{\theta_A \notin R_A \} \cup \{ \theta_B \notin R_B\} | \textbf{y}) \leq \mathbb{P}(\theta_A \notin R_A | \textbf{y}) + \mathbb{P}(\theta_B \notin R_B | \textbf{y}) \leq \alpha + \alpha = 2\alpha
\]

Thus:
\[
\mathbb{P}(\theta_A = \theta_B | \textbf{y}) \leq 2\alpha
\]
\end{proof}

\begin{proof}[Proof of proposition 3.6]
The three point of this proposition are: 

(i) The Bonferroni-type rectangle verifies Equation~\ref{eq:rectangle}:
\[
\mathbb{P}\left( (\theta_1, \dots, \theta_{d}) \in R(t^{\text{bonf}}_{\alpha}) |\mathbf{y} \right) \geq 1-\alpha. \]

(ii) The Naive rectangle based on $\pi(\theta|\mathbf{y})$ guarantees:

\[P(\theta \in R(t^{\text{naive}}_{\alpha}) |\mathbf{y}) \leq (1-\alpha) \]

(iii) The optimal quantile level $t_{\alpha}$ exists and is included inside the interval  $[\frac{\alpha}{d},\alpha]$

\textit{Proof of (i)}

By definition:

\[
R(t^{\text{bonf}}_{\alpha}) = \prod_{i=1}^d \left[ q_i\left(\frac{\alpha}{2d}\right),\; q_i\left(1-\frac{\alpha}{2d}\right) \right].
\]
The probability that the true parameter vector belongs to this region is:
\[
\mathbb{P}\left( (\theta_1, \dots, \theta_d) \in R(t^{\text{bonf}}_{\alpha}) |\mathbf{y}\right) = \mathbb{P}\left( \bigcap_{i=1}^{d} \left\{q_i\left(\frac{\alpha}{2d}\right) \leq \theta_i \leq q_i\left(1-\frac{\alpha}{2d}\right) \right\} |\mathbf{y}\right).
\]

Using the complement rule and Boole's inequality:
\[
\mathbb{P}\left( \bigcap_{i=1}^{d}  \left\{q_i\left(\frac{\alpha}{2d}\right) \leq \theta_i \leq q_i\left(1-\frac{\alpha}{2d}\right) \right\} |\mathbf{y}\right) = 1 - \mathbb{P}\left( \bigcup_{i=1}^{d}  \left\{\theta_i < q_i\left(\frac{\alpha}{2d}\right) \;\text{or}\; \theta_i > q_i\left(1-\frac{\alpha}{2d}\right) \right\} |\mathbf{y} \right).
\]

By definition of the quantile $q_i(\cdot)$, we have $\mathbb{P}(\theta_i < q_i(\frac{\alpha}{2d})|\mathbf{y}) \leq \frac{\alpha}{2d}$ and $\mathbb{P}(\theta_i > q_i(1-\frac{\alpha}{2d})|\mathbf{y}) \leq \frac{\alpha}{2d}$. Therefore:
\[
\mathbb{P}\left( (\theta_1, \dots, \theta_{d} ) \in R_{\text{bonf}}^{(1-\alpha)} |\mathbf{y}\right) \geq 1 - \sum_{i=1}^{d} \left( \frac{\alpha}{2d} + \frac{\alpha}{2d} \right) = 1 - \alpha.
\]

This proves that the Bonferroni-type rectangle is a credible region at a level $1-\alpha$, without independence assumption between the $\theta_i$. 

\textit{Proof of (ii)}

The Naive Rectangle is defined by : 
\[
R(t^{\text{naive}}_{\alpha}) = \prod_{i=1}^d \left[ q_i\left(\frac{\alpha}{2}\right),\; q_i\left(1-\frac{\alpha}{2}\right) \right].
\]

In particular we have : 
\[ \theta \in R(t^{\text{naive}}_{\alpha}) \implies \left\{q_1\left(\frac{\alpha}{2}\right) \leq \theta_1 \leq q_1\left(1-\frac{\alpha}{2}\right)\right\}.\]

Hence:

\[P(\theta \in R(t^{\text{naive}}_{\alpha}) |\mathbf{y}) \leq \mathbb{P}(\left[q_1\left(\frac{\alpha}{2}\right),q_1\left(1-\frac{\alpha}{2}\right)\right] | \mathbf{y})=1-\alpha.\]

\textit{Proof of (iii)}

The two bounds from the naive and Bonferroni-type quantile levels in (i) and (ii), combined with the fact that the probability of the quantile-based region only increases as the quantile level decreases, prove point (iii): the existence of a level $t_{\alpha} \in [\frac{\alpha}{d},\alpha]$ such that the probability of the rectangle set $R(t_{\alpha})$ is exactly $1 - \alpha$.
\end{proof}

\begin{proof}[Proof of Proposition 3.7]

We want to prove that he rectangle constructed by the BGHM algorithm in the case where $k$ is equal to $\lceil (1-\alpha) M \rceil$ provides an estimation of the optimal rectangle $R(t_{\alpha})$, as well as an estimation of $t_{\alpha}$: 
\[\hat{t}_{\alpha}=\frac{2(M+1-j^{\star})}{M+1}.\]

Results from \cite{besag_bayesian_1995} established that the rectangle defined in the Algorithm~\ref{algo:besag} is the smallest rectangle of this type containing at least $k$ of the $M$ samples. More precisely it implies that the proportion of samples inside this rectangle is between $\frac{k}{M}$ and $\frac{k+2d}{M}$. In this case if we define $k$ as a function of $M$ such that $k=\lceil=(1-\alpha)M\rceil$, then the associated rectangle $R_{M,k}$ given by the previous algorithm has the following property:
\[
\lim_{M \to \infty} \mathbb{P}(\theta\in R_{M,k})=1-\alpha.
\]

From this point we just have to prove that this hyperrectangle estimates a quantile-based rectangle to prove that this is an estimation of the optimal rectangle.

Given $M$ samples for parameter $\theta_i$, the marginal order statistics $\theta_i^{(1)}\le\cdots\le\theta_i^{(M)}$ provide empirical quantile estimates:
\[
\theta_i^{[r]}=\widehat{q}_i\!\left(\frac{r}{M+1}\right).
\]

Thus the hyperrectangle
\[
\prod_{i=1}^{d}\bigl[\theta_i^{[M+1-j^{\star}]},\ \theta_i^{[j^{\star}]}\bigr]
\]
estimates the theoretical region
\[
\widehat{R}(t)=\prod_{i=1}^{d}\bigl[\,\widehat{q}_i(t/2),\;\widehat{q}_i(1-t/2)\bigr],
\]
where the quantile level $t$ satisfies $t=\frac{2(M+1-j^{\star})}{M+1}$.

The previous result about the asymptotic probability of this rectangle tells us that this is an estimation of $R(t_{\alpha})$, and that an estimation of $t_{\alpha}$ is given by :
\[\hat{t}_{\alpha}=\frac{2(M+1-j^{\star})}{M+1}.\]
\end{proof}

\begin{proof}[Proof of Proposition 3.8]

Let $A, B \subset \mathbb{R}^d$ be two disjoint batches of points.  
For any integer $m \ge 1$, define:
\[
\mathcal{E}_m(A) = \text{the set of $m$ points in $A$ with the largest values of } S(\cdot, A),
\]
and similarly $\mathcal{E}_m(B)$ for $B$ and $\mathcal{E}_m(A \cup B)$ for $A \cup B$.

We want to prove that:
\[
\mathcal{E}_m(A \cup B) \subseteq \mathcal{E}_m(A) \cup \mathcal{E}_m(B).
\]

We prove the contrapositive: if a point $\theta \notin \mathcal{E}_m(A) \cup \mathcal{E}_m(B)$, then $\theta \notin \mathcal{E}_m(A \cup B)$.

Suppose $\theta \in A$ (the case $\theta \in B$ is symmetric). Since $\theta \notin \mathcal{E}_m(A)$, there exist at least $m$ points in $A$ with strictly larger scores than $\theta$ when computed within $A$:
\[
|\{ \theta' \in A : S(\theta', A) > S(\theta, A) \}| \ge m.
\]

Now consider the scores computed within the full set $A \cup B$. For any point $\theta' \in A$, its rank on each coordinate can only decrease or stay the same when adding points from $B$, since additional points can only insert themselves ahead in the ordering. Consequently, for any $\theta' \in A$:
\[
\min_i r_i(\theta', A \cup B) \le \min_i r_i(\theta', A), \quad \max_i r_i(\theta', A \cup B) \ge \max_i r_i(\theta', A).
\]

From the definition of $S$, this implies $S(\theta', A \cup B) \ge S(\theta', A)$. In particular, the $m$ points in $A$ with larger scores than $\theta$ in $A$ maintain larger scores than $\theta$ in $A \cup B$.

Thus, there are at least $m$ points in $A \cup B$ with scores strictly larger than $S(\theta, A \cup B)$, meaning $\theta \notin \mathcal{E}_m(A \cup B)$.
\end{proof}

\begin{proof}[Proof of Theorem 4.1]

We suppose that the model is under a Gaussian likelihood with zero mean, with a prior on the covariance matrix that is an inverse-Wishart distribution where the eigenvalues of $\boldsymbol{\boldsymbol{\Sigma}}_0$ are bounded below by a strictly positive constant. In this situation we define the strict upper-triangular vectorization $\textup{vech}_0 : \mathbb{R}^{p \times p} \to \mathbb{R}^{p(p-1)/2}$ of the corresponding correlation matrix $\boldsymbol{C}$.

We want to prove that in this case we have:
\[
\sqrt{n}\,\textup{vech}_0(\boldsymbol{C} - \boldsymbol{C}_n) | \mathbf{X}_n 
\;\stackrel{\mathrm{TV}}{\approx}\; 
\mathrm{N}_{p(p-1)/2}\bigl( \mathbf{0}, V \bigr),
\]

where $\boldsymbol{C}_n$ is the empirical correlation matrix derived from $\boldsymbol{S}_n$.

Define the transformation $g: \mathbb{R}^{p(p+1)/2} \to \mathbb{R}^{p(p-1)/2}$ mapping the unique elements of $\boldsymbol{\boldsymbol{\Sigma}}$ to the strict upper-triangular elements of the correlation matrix:
\[
g(\text{vech}(\boldsymbol{\boldsymbol{\Sigma}})) = \text{vech}_0(\boldsymbol{C}),
\]
where $\boldsymbol{C} = \boldsymbol{D}^{-1/2}\boldsymbol{\boldsymbol{\Sigma}}\boldsymbol{D}^{-1/2}$ and $\boldsymbol{D} = \text{diag}(\boldsymbol{\boldsymbol{\Sigma}})$.

The function $g$ is differentiable at $\boldsymbol{\boldsymbol{\Sigma}}$ because the diagonal entries are supposed to be bounded away from zero. Let $\nabla g$ denote its Jacobian matrix evaluated at $\boldsymbol{\boldsymbol{\Sigma}}$.

From the BvM for covariance matrices, we have
\[
\sqrt{n}\,\text{vech}(\boldsymbol{\boldsymbol{\Sigma}} - \boldsymbol{S}_n) | \mathbf{X}_n \stackrel{\mathrm{TV}}{\approx} \mathrm{N}(\boldsymbol{0},\, \boldsymbol{V}_{\boldsymbol{\Sigma}}),
\]
with $\boldsymbol{V}_{\boldsymbol{\Sigma}} = 2\boldsymbol{B}_p^\top(\boldsymbol{S}_n \otimes \boldsymbol{S}_n)\boldsymbol{B}_p$.

Applying the delta method (which preserves convergence in total variation for differentiable transformations), and using the linearity of $\textup{vech}$ and $\textup{vech}_0$ yields
\[
\sqrt{n}\,(g(\text{vech}(\boldsymbol{\boldsymbol{\Sigma}})) - g(\text{vech}(\boldsymbol{S}_n))) | \mathbf{X}_n \stackrel{\mathrm{TV}}{\approx} \mathrm{N}(\boldsymbol{0},\, \nabla g \, \boldsymbol{V}_{\boldsymbol{\Sigma}} \, \nabla g^\top),
\]

which results in:

\[
\sqrt{n}\,(\text{vech}_0(\boldsymbol{C}) - \text{vech}_0(\boldsymbol{C}_n)) | \mathbf{X}_n \stackrel{\mathrm{TV}}{\approx} \mathrm{N}(\boldsymbol{0},\, \nabla g \, \boldsymbol{V}_{\boldsymbol{\Sigma}} \, \nabla g^\top),
\]

where $\boldsymbol{C}_n$ is the empirical correlation matrix.
\end{proof}

\begin{proof}[Proof of Proposition 4.2]

We aim to prove that $\widehat{A}_n$ converges in probability to $A_0$ as $n \to \infty$, i.e., the support estimator is consistent for any sequence $\alpha_n \to 0$ that respects the Gaussian tail decay implied by the BvM.

The Bernstein-von Mises Theorem for correlation matrices vectorization~\ref{theorem:bmv} guarantees that, under standard regularity conditions, the posterior distribution $\pi(\text{vech}_0(\boldsymbol{C}) | \mathbf{X}_n)$ contracts at rate $1/\sqrt{n}$ around the true $\text{vech}_0(\boldsymbol{C}_0)$.
In particular, for every $(i,j)$:
\begin{itemize}
    \item If $\boldsymbol{C}_{0,ij} \neq 0$, the posterior of $\boldsymbol{C}_{ij}$ concentrates around $\boldsymbol{C}_{0,ij}$. The credible interval $I_{ij}(t_{\alpha_n})$ is asymptotically centered around $\boldsymbol{C}_{0,ij}$ with width of order $O_p(1/\sqrt{n})$. Provided $\alpha_n \to 0$ sufficiently slowly to maintain power against local alternatives, the interval excludes $0$ with probability tending to one, as $\boldsymbol{C}_{0,ij}$ is bounded away from zero.

    \item If $\boldsymbol{C}_{0,ij} = 0$, the posterior of $\boldsymbol{C}_{ij}$ is asymptotically normal with mean zero and variance of order $1/n$. The sequence $\alpha_n \to 0$ is chosen to decay fast enough to control false discoveries (shrinking the interval relative to fixed thresholds) but slow enough to respect the Gaussian contraction rate. Consequently, the interval contracts toward zero, ensuring $0 \in I_{ij}(t_{\alpha_n})$ with probability tending to one.
\end{itemize}
\end{proof}

\bibliographystyle{plainnat} 
\bibliography{references}

\end{document}